\documentclass[12pt]{article}
\usepackage{enumitem}
\usepackage{mathtools}
\usepackage{xr}

\usepackage{url}
\usepackage{amsfonts}
\usepackage{amsmath,amssymb,color}
\usepackage{graphicx}
\usepackage{enumitem}
\usepackage{multirow}
\usepackage{tabularx}
 
\usepackage{bm}
\usepackage{verbatim}
\usepackage{float}
\usepackage[toc,page]{appendix}
\usepackage{apacite}
\usepackage{natbib}
\usepackage{adjustbox}
\usepackage{rotating}
\usepackage{subcaption}
\usepackage{amsthm}
\usepackage{booktabs,dcolumn,caption}
\usepackage{comment}
\usepackage{mathabx}

\usepackage{array}

\usepackage[ruled]{algorithm2e}
\usepackage{amsthm}
\usepackage{hyperref}
\usepackage{cleveref}

\newcolumntype{L}[1]{>{\raggedright\let\newline\\\arraybackslash\hspace{0pt}}p{#1}}
\newcolumntype{C}[1]{>{\centering\let\newline\\\arraybackslash\hspace{0pt}}p{#1}}
\newcolumntype{R}[1]{>{\raggedleft\let\newline\\\arraybackslash\hspace{0pt}}p{#1}}

\newcommand{\pr}{{\text{pr}}}
\newcommand{\bfs}{{\boldsymbol s}}
\newcommand{\bfS}{{\mathbf S}}
\newcommand{\bfY}{{\mathbf Y}}
\newcommand{\bfZ}{{\mathbf Z}}
\newcommand{\bfy}{{\boldsymbol y}}

\newcommand{\bfalpha}{{\boldsymbol{\alpha}}}

\newcommand{\bbb}{{\boldsymbol \beta}}

\newcommand{\YY}{\mbox{$\mathbf Y$}}

\newcommand{\bfdelta}{{\boldsymbol \delta}}
\newcommand{\bfkappa}{{\boldsymbol \kappa}}
\newcommand{\bfmu}{{\boldsymbol \mu}}

\usepackage{setspace}

\theoremstyle{plain}
\newtheorem{theorem}{Theorem}
\newtheorem{lemma}[theorem]{Lemma}

\newtheorem{definition}{Definition}

\newtheorem{proposition}{Proposition}
\newtheorem{remark}{Remark}

\newtheorem{assumption}{Assumption}

\title{\LARGE A Note on Ising Network Analysis with Missing Data}
\author{Siliang Zhang and Yunxiao Chen}

\date{}

\begin{document}
\setstretch{1.4}
\maketitle

\begin{abstract}

The Ising model has become a popular psychometric model for analyzing item response data. The statistical inference of the Ising model is typically carried out via a pseudo-likelihood, as the standard likelihood approach suffers from a high computational cost when there are many variables (i.e., items). Unfortunately, the presence of missing values can hinder the use of 
pseudo-likelihood, and a listwise deletion approach for missing data treatment may introduce a substantial bias into the estimation and sometimes yield misleading interpretations. This paper proposes a conditional Bayesian framework for Ising network analysis with missing data, which integrates a pseudo-likelihood approach with iterative data imputation. 
An asymptotic theory is established for the method. 
Furthermore, a computationally efficient {P{\'o}lya}-Gamma data augmentation procedure is proposed to streamline the 
sampling of model parameters. The method's performance is shown through  simulations and a real-world application to data on major depressive and generalized anxiety disorders from the National Epidemiological Survey on Alcohol and Related Conditions (NESARC).

\end{abstract}	

\section{Introduction}

Recent years have witnessed the emergence of network psychometrics \citep{vanderMaas2006,borsboom2008psychometric,marsman2022guest}, a family of statistical graphical models and related inference procedures, for analyzing and interpreting the dependence structure in psychometric data. These models embed psychometric items as nodes in an undirected or directed network (i.e., graph) and visualize their interrelationships through the network edges, which represent certain probabilistic conditional dependencies. Network psychometric methods concern the learning of the network structure. They have been developed under various 
settings, including undirected graphical models for cross-sectional data \citep[]{epskamp2018gaussian,burger2022reporting}, directed networks for longitudinal data \citep[]{gile2017analysis,borsboom2021network,ryan2022challenge}, {and extended networks with latent variables for time-series data or panel data \citep[]{epskamp2020psychometric}.} These methods have received wide applications in education \citep[]{sweet2013hierarchical,willcox2017network,koponen2019using,siew2020applications,simon2021network}, 
psychology \citep[]{burgess1999memory,van2017network,fried2017mental,epskamp2018estimating,borsboom2021network}, and health sciences \citep[]{luke2007network,brunson2018applications,mkhitaryan2019network,kohler2022using}.  

Analyzing cross-sectional binary item response data with the Ising model \citep[]{Ising1925} is common in network psychometric analysis. This analysis is typically performed based on a conditional likelihood \citep[]{besag1974spatial} because the standard likelihood function 
is computationally infeasible when involving many variables. In this direction, Bayesian and frequentist methods have been developed, where sparsity-inducing priors or penalties are combined with the conditional likelihood for learning a sparse network structure \citep[]{yuan2007model,mazumder2012graphical,van2014new,epskamp2018tutorial,li2019graphical,marsman2022objective}. Besides, the Ising model is shown to be closely related to Item Response Theory (IRT) models \citep{holland1990dutch,anderson2007log}. The log-multiplicative association models \citep{anderson2007log}, which are special cases of the Ising model, can be used as item response theory models and yield very similar results as IRT models. Furthermore, the Ising model and the conditional likelihood have been used for modeling the local dependence structure in locally dependent IRT models \citep{ip2002locally,chen2018robust}. 

Due to its construction, the conditional likelihood does not naturally handle data with missing values, despite the omnipresence of missing data in psychometric data. 
To deal with missing values in an Ising network analysis, listwise deletion \citep[]{haslbeck2017predictable,fried2020using} and single imputation \citep[e.g.,][]{huisman2009imputation,armour2017network,lin2020association} are typically performed, which arguably may not be the best practice. 
In particular, it is well-known that
listwise deletion is statistically inefficient and requires 
the Missing Completely At Random (MCAR) assumption \citep[]{little2019statistical}  to ensure consistent estimation. 
Moreover, a  na\"ive imputation procedure, such as mode imputation, likely introduces bias into parameter estimation. A sophisticated imputation procedure must be developed to ensure statistical validity and computational efficiency.

In this note, we propose an iterative procedure for learning an Ising network. The proposed procedure combines iterative imputation via Full Conditional Specification \citep[FCS;][]{liu2014stationary,buurenFlexibleImputationMissing2018} and Bayesian estimation of the Ising network. We show that the FCS leads to estimation consistency when the conditional models are chosen to take logistic forms. In terms of computation, we propose a joint {P{\'o}lya}-Gamma augmentation procedure by extending the 
{P{\'o}lya}-Gamma augmentation procedure for logistic regression \citep[]{polson2013bayesian}. It allows us to efficiently sample parameters of the Ising model. Simulations are conducted to compare the proposed procedure with estimations based on the listwise deletion and single imputation. Finally, the proposed procedure and a complete-case analysis are applied to study the network of Major Depressive Disorder (MDD) and Generalised Anxiety Disorders (GAD) based on data from the National Epidemiological Survey on Alcohol and Related Conditions \citep[NESARC;][]{grant2003source}. In this analysis, data missingness is mainly due to two screening items for GAD. That is, a respondent's responses to the rest of the MDD items are missing if they answered ``no" to both screening items. This missing mechanism is Missing at Random \citep[MAR;][]{little2019statistical}. The complete-case analysis of missing data caused by screening items is known to be problematic in the literature of network psychometrics \citep{borsboom2017false,mcbride2023quantifying}. Our Bayesian estimate of the edge coefficient between the two screening items is negative based on the complete cases, which can be seen as a result of Berkson’s paradox \citep{de2021psychological}. In contrast, the proposed method makes use of all the observed data entries and obtains a positive estimate of this edge coefficient. 
An identifiability result about the Ising model under this special missing data setting {in the Appendix}, the item content, and a simulation study mimicking this setting suggest that the result given by the proposed method is sensible.

\section{Proposed Method}\label{sec:model_method}
 
\subsection{Ising Model}

Consider a respondent answering $J$ binary items. Let $\YY = (Y_1, ..., Y_J)^\top \in \{0,1\}^J$ be a binary random vector representing the respondent's responses. We say $\YY$ follows an Ising model if its probability mass function satisfies  
\begin{equation}\label{eq:ising_model}
  \begin{aligned}
    P(\YY = \bfy \mid \bfS) &= \frac{1}{c(\bfS)}\exp\left[\frac{1}{2}\bfy^\top\bfS\bfy\right]  =\frac{1}{c(\bfS)}\exp\left[\sum_{j=1}^J s_{jj}y_{j}/2+\sum_{j=1}^{J-1}\sum_{k=j+1}^Js_{jk}y_{j}y_k\right],
  \end{aligned}
\end{equation}
where $\bfS = (s_{ij})_{J\times J}$ is a $J$ by $J$ symmetric matrix that contains parameters of the Ising model
and 
$$c(\bfS) = \sum_{\bfy\in\{0,1\}^J}\exp\left[\sum_{j=1}^J s_{jj}y_{j}/2+\sum_{j=1}^{J-1}\sum_{k=j+1}^Js_{jk}y_{j}y_k\right]$$
is a normalizing constant. The parameter matrix $\bfS$ encodes a network with the $J$ items being the nodes. More specifically, an edge is present between nodes $i$ and $j$ if and only if the corresponding entry $s_{ij}$ is nonzero.
If an edge exists between nodes $i$ and $j$, then $Y_i$ and $Y_j$ are conditionally dependent given the rest of the variables. Otherwise,  the two variables are conditionally independent.


In Ising network analysis, the goal is to estimate the parameter matrix $\bfS$. 
The standard likelihood function is computationally intensive when $J$ is large, as it 
requires computing a normalizing constant $c(\bfS)$ which involves a summation of all the $2^J$ response patterns. To address this computational issue, \cite{besag1975statistical} proposed a conditional likelihood which is obtained by aggregating the conditional distributions of $Y_j$ given $\YY_{-j} = (Y_1, ..., Y_{j-1}, Y_{j+1}, ..., Y_J)^\top$, for $j=1, ..., J$, where the conditional distribution of $Y_j$ given $\YY_{-j}$ takes a logistic regression form. More precisely, the conditional likelihood with one observation $\bfy$ is defined as 
 
\begin{equation}\label{eq:ising_model_pseudo}
  \begin{aligned}
    p^*(\bfy\mid \bfS) &= \prod_{j=1}^J p(y_j\mid \bfy_{-j},\bfS)= \prod_{j=1}^J\frac{\exp\left[(s_{jj}/2 + \sum_{k\neq j} s_{jk}y_k)y_j\right]}{1+\exp\left(s_{jj}/2 + \sum_{k\neq j} s_{jk}y_k\right)}. 
  \end{aligned}
\end{equation}
A disadvantage of  the conditional likelihood is that it requires a fully observed dataset because missing values cannot be straightforwardly marginalized out from \eqref{eq:ising_model_pseudo}. In what follows, we discuss how missing data can be treated in the conditional likelihood.

\subsection{Proposed Method}

Consider a dataset with $N$ observations. Let $\Omega_{j} \subset \{1, ..., N\}$ be the subset of observations whose data on item $j$ are missing. For each observation $i$ and item $j$, $y_{ij}$ denotes the observed response if $i \notin \Omega_j$, and otherwise, $y_{ij}$ is missing. Thus, the observed data include $\Omega_j$ and $y_{ij}$, for $i\in \{1, ..., N\}\setminus\Omega_j$ and $j = 1, ..., J$. 

The proposed procedure iterates between two steps -- (1) imputing
the missing values of $y_{ij}$ for $i \in \Omega_j$, $j=1, ..., J$, achieved via a full conditional specification, 
and (2) sampling the posterior distribution of $\bfS$ given the most recently imputed data. 
Let $t$ be the current iteration number. Further, let 
 $\bfy^{(t-1)}_i = (y_{i1}^{(t-1)}, ..., y_{iJ}^{(t-1)})^\top, i=1, ..., N$, be imputed data from the previous iteration, where $y^{(t-1)}_{ij} = y_{ij}$ for $i \notin \Omega_j$ and  $y^{(t-1)}_{ij}$ is imputed in the $(t-1)$th iteration for $i \in \Omega_j$. 
For the $t$th iteration, the imputation and sampling steps are described as follows.

\paragraph{Imputation.} We initialize the imputation in the $t$th iteration 
with the previously imputed data set $(y_{i1}^{(t-1)},\ldots,y_{iJ}^{(t-1)})$, $i=1, ..., N$. Then, we run a loop over all the items, $j = 1, ..., J$. In step $j$ of the loop, we impute $y_{ij}$ for all $i \in \Omega_j$, 
given the most recently imputed data 
 $(y_{i1}^{(t)},\ldots,y_{i,j-1}^{(t)},y_{ij}^{(t-1)},\ldots,y_{iJ}^{(t-1)})$, $i = 1, ..., N$. We then obtain 
 an updated data set $(y_{i1}^{(t)},\ldots,y_{i,j}^{(t)},y_{i,j+1}^{(t-1)},\ldots,y_{iJ}^{(t-1)})$ by incorporating the newly imputed values for $y_{ij}$.

The imputation of each variable $j$ is based on the conditional distribution of $Y_j$ given $\YY_{-j}$. Under the Ising model, this conditional distribution takes a logistic regression form. For computational reasons to be discussed in the sequel, we introduce an auxiliary parameter vector $\bbb_j = (\beta_{j1}, ..., \beta_{jJ})^\top$ as coefficients in the logistic regression, instead of directly using $\bfS$ from the previous iteration to sample the missing $y_{ij}$s. Unlike the constraint of $s_{ij} = s_{ji}$ in the symmetric matrix $\bfS$, no constraints are imposed on $\boldsymbol{\beta}_j$, $j=1, ..., J$, which makes the sampling computationally efficient; see discussions in Section~\ref{subsec:comp}. 
The imputation of variable $j$ consists of the following two steps: 

\begin{enumerate}
    \item Sample auxiliary parameter vector $\bbb_j^{(t)}$ from the posterior distribution 
    \begin{equation}\label{eq:post_beta}
        p^{(t,j)}(\bbb_j) ~\propto~ \pi_j(\bbb_j)\prod_{i=1}^N \frac{\exp\left[(\beta_{jj}/2+\sum_{k\neq j}\beta_{jk}y_{ik}^{(t,j-1)})y_{ij}^{(t,j-1)}\right]}{1+\exp(\beta_{jj}/2+\sum_{k\neq j}\beta_{jk}y_{ik}^{(t,j-1)})},
    \end{equation}
where $\pi_j(\bbb_j)$ is the prior distribution for the auxiliary parameters $\bbb_j$.

    \item Sample $y_{ij}^{(t)}$ for each $i \in \Omega_j$ from a Bernoulli distribution with success probability 
\begin{equation}\label{eq:sample_y}
   \frac{\exp(\beta_{jj}^{(t)}/2+\sum_{k\neq j}\beta_{jk}^{(t)}y_{ik}^{(t,j-1)})}{1+\exp(\beta_{jj}^{(t)}/2+\sum_{k\neq j}\beta_{jk}^{(t)}y_{ik}^{(t,j-1)})}. 
\end{equation}

\end{enumerate}
After these two steps, we obtain 
$(y_{i1}^{(t)},\ldots,y_{i,j}^{(t)},y_{i,j+1}^{(t-1)},\ldots,y_{iJ}^{(t-1)})$ by incorporating the newly imputed values for $y_{ij}, i\in\Omega_j$.
We emphasize that only the missing values are updated. For $i  \notin \Omega_j$, 
$y_{ij}^{(t)}$ is always the observed value of $y_{ij}$. After the loop over all the items, 
we have the imputed data set $(y_{i1}^{(t)},\ldots,y_{iJ}^{(t)})$ as the output from this imputation step.

\paragraph{Sampling $\bfS$.} Given the most recently imputed data $\bfy_i^{(t)}$, $i=1, ..., N$, 
update $\bfS^{(t)}$ by sampling from the pseudo-posterior distribution
\begin{equation}\label{eq:s_cond}
    p(\bfS\mid\bfy_1^{(t)},\ldots,\bfy_N^{(t)})~\propto~ \pi(\bfS)\prod_{i=1}^N p^*(\bfy_i^{(t)} \mid \bfS),
\end{equation}
where $\pi(\bfS)$ is the prior distribution for the Ising  parameter matrix $\bfS$ and recall that $\prod_{i=1}^N p^*(\bfy_i^{(t)} \mid \bfS)$ is the conditional likelihood. 

\begin{figure}
    \centering
    \includegraphics[scale=0.45]{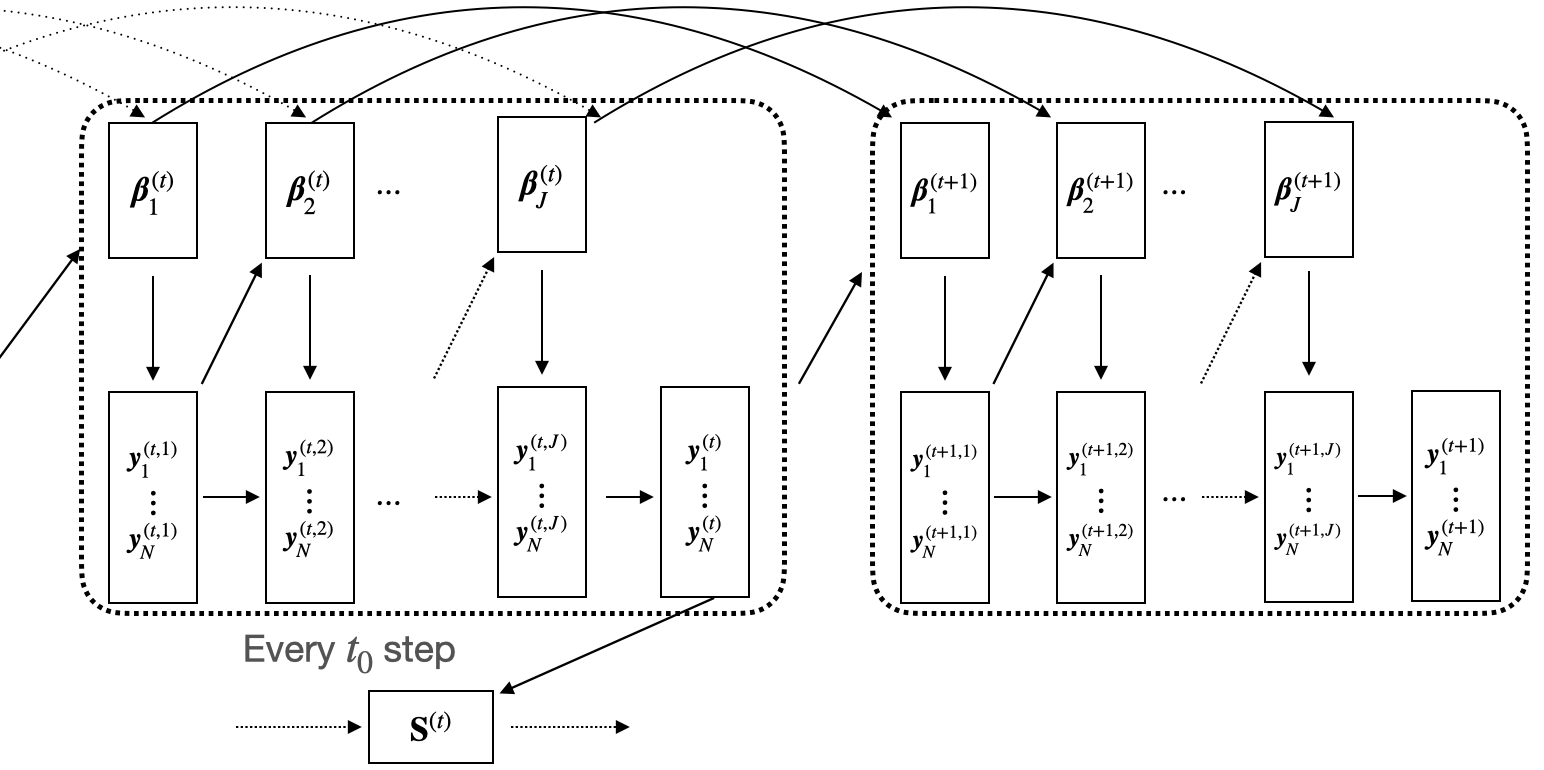}
    \caption{Flow chart of the updating rule for the proposed method.}
    \label{fig:diagram}
\end{figure}

\begin{algorithm}
  \SetKwInOut{Output}{Output}
  \caption{Ising Network Analysis with Iterative Imputation}\label{alg:pseudo_imputation}
  \setstretch{1.2}
  \KwData{Given observed data, initial values for the Ising model parameters $\bfS^{(0)}$, randomly imputed missing data $\bfy_1^{(0)},\ldots,\bfy_N^{(0)}$, MCMC length $T$, burn-in size $T_0$, thinning steps size $t_0$. Let auxiliary parameters $\bbb_j^{(0)}=\bfs_j^{(0)},j=1,\ldots,J$.}
  \For{each iteration $t = 1$ to $T$}{
    
    \For{each $j = 1$ to $J$}{      
      
      {Sample auxiliary parameter vector $\bbb_j^{(t)}$ from $p^{(t,j)}(\bbb_j).$}

      {Sample $y_{ij}^{(t)}$ for each $i\in\Omega_j$ from the Bernoulli distribution given in \eqref{eq:sample_y}.} 
    }
    \If{$t>T_0$ and $t$ is a multiple of $t_0$}{
    Sample  $\bfS^{(t)}$  from $p(\bfS\mid \bfy_1^{(t)},\ldots,\bfy_N^{(t)})$ given in \eqref{eq:s_cond}.
    }
  } 
  \Output{$\hat\bfS = \frac{1}{M- M_0}\sum_{t \in \{(M_0+1)t_0, ..., Mt_0\}}\bfS^{(t)},$ where $M=\lfloor T/t_0\rfloor$ and $M_0=\lfloor T_0/t_0\rfloor$.}
\end{algorithm}
Figure~\ref{fig:diagram} visualizes the steps performed in the proposed method. Note that it is unnecessary to sample the parameter matrix $\bfS$ during the burn-in period and in 
every iteration after the burn-in period; thus, we employ a thinning step after the burn-in period. This is done to both decrease computational cost and reduce the auto-correlation in the imputed data. 
Moreover, we outline the proposed algorithm in Algorithm~\ref{alg:pseudo_imputation}. The final estimate of $\bfS$ is obtained by averaging all the $\bfS^{(t)}$ obtained after the burn-in period.
The computational details, including the sampling of auxiliary parameters and Ising parameter matrix and discussions of the computational complexity, are given in Section~\ref{subsec:comp}.

We remark that our method imputes the missing variables one by one for each observation. This method is chosen because simultaneously imputing all the missing variables is typically computationally infeasible, especially when some observation units have many missing values. Simultaneous imputation requires evaluating the joint distribution of the missing variables given the observed ones, whose computational complexity grows exponentially with the number of missing values. In contrast, the proposed method is based on unidimensional conditional distributions, which is computationally more feasible. We also note that the proposed method has several variants that should also work well. These variants are discussed in Section~\ref{sec:conclusion}.


\subsection{Statistical Consistency}\label{par:consistency}
As our method is not a standard Bayesian inference procedure, we provide  an asymptotic theory under the frequentist setting to justify its validity. In particular, we show that 
 the $\bfS$ parameter sampled from the pseudo-posterior distribution converges to the true parameter $\bfS_0$,  under the assumptions that the Ising model is correctly specified and the data are MAR.

Consider one observation with a complete data vector $\YY = (Y_1, ..., Y_J)^\top$. Further, let $\bfZ = (Z_1, ..., Z_J)^\top$ be a vector of missing indicators, where $Z_{ij} = 1$ if $Y_{ij}$ is observed and  $Z_{ij} = 0$ otherwise. We further let  $\YY_{obs} = \{Y_j: Z_j = 1, j = 1, ..., J\}$ and $\YY_{mis} = \{Y_j: Z_j = 0, j = 1, ..., J\}$ be the observed and missing entries of $\YY$, respectively. Consider the joint distribution of observable data $(\YY_{obs}, \bfZ)$, taking the form

\begin{equation}\label{eq:joint} 
   P(\YY_{obs} = \bfy_{obs}, \bfZ = \mathbf z\mid \bfS, \boldsymbol{\phi}) = \sum_{y_j:z_j=0} \left( \exp\left(\bfy^\top\bfS\bfy/2\right)/{c(\bfS)} \right) q(\mathbf z \mid \bfy,\boldsymbol{\phi}),
\end{equation}
where $\exp\left(\bfy^\top\bfS\bfy/2\right)/{c(\bfS)}$ is the distribution of $\YY = \bfy$ under the Ising model,  
 $q(\mathbf z \mid \bfy,\boldsymbol{\phi})$ denotes the conditional probability of $\bfZ = \mathbf z$ given $\YY = \bfy$, and $\boldsymbol \phi$ denotes the unknown parameters of this distribution. The MAR assumption, also known as the ignorable missingness assumption,  means that the conditional distribution $q(\mathbf z \mid \bfy,\boldsymbol{\phi})$ depends on $\bfy$ only through the observed entries, i.e., $q(\mathbf z \mid \bfy,\boldsymbol{\phi}) = q(\mathbf z \mid \bfy_{obs},\boldsymbol{\phi})$. In that case, \eqref{eq:joint} can be factorized as 
\begin{equation}\label{eq:joint2} 
   P(\YY_{obs} = \bfy_{obs}, \bfZ = \mathbf z\mid \bfS, \boldsymbol{\phi}) = q(\mathbf z \mid \bfy_{obs},\boldsymbol{\phi}) \times \left(\sum_{y_j:z_j=0}  \exp\left(\bfy^\top\bfS\bfy/2\right)/{c(\bfS)}   \right).
\end{equation}
Consequently, the inference of $\bfS$ does not depend on the unknown distribution 
 $q(\mathbf z \mid \bfy,\boldsymbol{\phi})$. 


As shown in \cite{liu2014stationary}, the MAR assumption, together with additional regularity conditions, ensures that the iterative imputation of the missing responses converges to the imputation distribution under a standard Bayesian procedure as the number of iterations and the sample size $N$ go to infinity. A key to this convergence result is the compatibility of the conditional models in the imputation step -- the logistic regression models are compatible with the Ising model as a joint distribution. 
The validity of the imputed samples further ensures the consistency of the estimated Ising parameter matrix. We summarize this result in Theorem~\ref{thm:params_consist}.

\begin{theorem}\label{thm:params_consist}
  Assume the following assumptions hold: 1) The Markov chain for missing data, generated by the iterative imputation algorithm Algorithm~\ref{alg:pseudo_imputation}, is positive Harris recurrent and thus admits a unique stationary distribution; 2) The missing data process is ignorable; 3) A regularity condition holds for prior distributions of Ising model parameters and auxiliary parameters, as detailed in the supplementary material. 
  Let $\pi_N^*(\bfS)$ be the posterior density of $\bfS$ implied by the stationary distribution of the proposed method. Given the true parameters $\bfS_0$ for the Ising model and any $\varepsilon >0,$ we have $\pi_N^*(\bfS)$ concentrates at $\bfS_0,$
  \begin{equation}
     \int_{B_\varepsilon(\bfS_0)}\pi_N^*(\bfS)d\bfS \rightarrow 1,
  \end{equation}
  in probability as $N\rightarrow\infty$. $B_\varepsilon(\bfS_0)=\{\bfS:\Vert\bfS-\bfS_0\Vert<\varepsilon\}$ is the open ball of radius $\varepsilon$ at $\bfS_0$.
\end{theorem}

We provide intuitions about this consistency result. Suppose that the data are generated by an Ising model. The iterative imputation method ensures that the parameters of the logistic regressions are close to those implied by the true Ising model, and thus, the conditional distributions we use to impute the missing values are close to those under the true model.
This further guarantees that the joint distribution of the imputed data given the observed ones is close to that under the true Ising model, and consequently, the Ising model parameters we learn from the imputed data are close to those of the true model.

\subsection{Computational Details}\label{subsec:comp}

In what follows, we discuss the specification of the prior distributions and the sampling of auxiliary parameters $\boldsymbol \beta_j$ and Ising model parameters $\bfS$.

\paragraph{Sampling $\boldsymbol{\beta}_j$.} 

We set independent mean-zero normal priors for entries of $\boldsymbol{\beta}_j$. For the intercept parameter $\beta_{jj}$, we use a weakly informative prior by setting the variance to 100. For the slope parameters $\beta_{jk}$, $k\neq j$, we set a more informative prior by setting the variance to be 1, given that these parameters correspond to the off-diagonal entries of $\bfS$, which are sparse and typically do not take extreme values. 
The sampling of the auxiliary parameters $\bbb_j$, following \eqref{eq:post_beta}, is essentially a standard Bayesian logistic regression problem. We achieve it by a Markov chain Monte Carlo (MCMC) sampler called the {P{\'o}lya}-Gamma sampler \citep[]{polson2013bayesian}. 

To obtain $\boldsymbol{\beta}^{(t)}_j$, this sampler starts with $\boldsymbol{\beta}^{(t-1)}_j$ from the previous step. It constructs an MCMC transition kernel by a data argumentation trick. More precisely, the following two steps are performed.
\begin{enumerate}
    \item Given $\boldsymbol{\beta}^{(t-1)}$, independently sample $N$ augmentation variables, each from a  {P{\'o}lya}-Gamma distribution \citep[]{barndorff1982normal}.
    \item Given the $N$ augmentation variables, sample $\boldsymbol{\beta}^{(t)}$ from a $J$-variate normal distribution. 
\end{enumerate}
The details of these two steps are given in the supplementary material, including the forms of the {P{\'o}lya}-Gamma distributions and the mean and covariance matrix of the $J$-variate normal distribution. We choose the {P{\'o}lya}-Gamma sampler because it is very easy to construct and computationally efficient. 
It is much easier to implement than Metropolis-Hastings samplers which  often need tuning to achieve good performance.

We comment on the computational complexity of the sampling of $\boldsymbol{\beta}_j$. The sampling from the {P{\'o}lya}-Gamma distribution has a complexity $O(NJ)$, and the sampling from the $J$-variate normal distribution has a complexity of $O(NJ^2)+O(J^3)$. Consequently, a loop of all the  $\bbb_j, j=1, ..., J$, has a complexity of 
$O((N+J)J^3)$.


\paragraph{Sampling $\bfS$.} Since $\bfS$ is a symmetric matrix, we reparametrize it by vectorizing its off-diagonal entries (including the diagonal entries). Specifically, the reparameterization is done by half-vectorization, denoted by $\bfalpha = \text{vech}(\bfS) = (s_{11}, ..., s_{J1}, s_{22}, ..., s_{J2}, ..., s_{JJ})^\top \in \mathbb R^{J(J+1)/2}$. It is easy to see that $\text{vech}(\cdot)$ is a one-to-one mapping between $\mathbb R^{J(J+1)/2}$ and $J\times J$ symmetric matrices. Therefore, we impose a prior distribution on $\bfalpha$ and sample $\bfalpha^{(t)}$ in the $t$th iteration when $\bfS$ is sampled. Then we let $\bfS^{(t)} = \text{vech}^{-1}(\bfalpha^{(t)})$.

Recall that a thinning step is performed, and $t_0$ is the gap between two samples of $\bfS$. Let $t$ be a multiple of $t_0$ and $\bfalpha^{(t-t_0)} = \text{vech}(\bfS^{(t-t_0)})$ be previous value of $\bfalpha$. The sampling of $\bfalpha^{(t)}$ is also achieved by a {P{\'o}lya}-Gamma sampler, which involves the following two steps similar to the sampling of $\boldsymbol{\beta}_j$. 
\begin{enumerate}
    \item Given $\bfalpha^{(t-t_0)}$, independently sample $NJ$ augmentation variables, each from a {P{\'o}lya}-Gamma distribution. 
    \item Given the $NJ$ augmentation variables, sample $\bfalpha^{(t)}$ from a $J(J+1)/2$-variate normal distribution. 
\end{enumerate}
The details of these two steps are given in the supplementary material. We note that the computational complexity of sampling the $NJ$ augmentation variables is $O(NJ^2)$, and that of sampling $\bfalpha^{(t)}$ is $O(NJ^5)+O(J^6)$, resulting in an overall complexity $O((N+J)J^5)$. Comparing the complexities of the imputation and sampling $\bfS$ steps, we notice that the latter is computationally much more intensive.
This is the reason why we choose to impute data by introducing auxiliary parameters $\boldsymbol{\beta_j}$s rather than using Ising network parameters $\bfS$ so that the iterative imputation mixes much faster in terms of the computation time. In addition, we only sample  $\bfS$ every $t_0$ iterations for a reasonably large $t_0$ to avoid a high computational cost and also reduce the auto-correlation between the imputed data.

We remark that \cite{marsman2022objective} considered a similar Ising network analysis problem based on fully observed data, in which they proposed a Bayesian inference approach based on a spike-and-slab prior to learning $\bfS$. Their Bayesian inference is also based on a {P{\'o}lya}-Gamma sampler. However, they combined Gibbs sampling with a 
{P{\'o}lya}-Gamma sampler, updating one parameter in $\bfS$ at a time. This Gibbs scheme often mixes slower than the joint update of $\bfS$ as in the proposed method and, thus, is computationally less efficient. The proposed {P{\'o}lya}-Gamma sampler may be integrated into the framework of \cite{marsman2022objective} to improve their computational efficiency.

\section{Numerical Experiments} 

We illustrate the proposed method and show its power via simulation studies and a real-world data application. In Section 3.1, we conduct two simulation studies, evaluating the proposed method under two MAR scenarios, one of which involves missingness due to screening items. A further simulation study is undertaken, applying our method to a 15-node Ising model governed by the MCAR mechanism. Detailed exposition of this study can be found in the supplementary materials.

\subsection{Simulation}\label{sec:6node} 

\paragraph{Study I} 
We generate data from an Ising model with $J=6$ variables. Missing values are generated under an MAR setting that is not MCAR. The proposed method is then compared with Bayesian inference based on (1) listwise deletion and (2) a single imputation, where the single imputation is based on the imputed data from the $T$th iteration of Algorithm~\ref{alg:pseudo_imputation}, recalling that $T_0$ is the burn-in size.

We configure the true parameter matrix $\bfS_0$ as follows. Since $\bfS_0$ is a symmetric matrix, we only need to specify its upper triangular matrix and then the diagonal entries. For the upper triangular entries (i.e., $s_{jl}$, $j<l$), 
we randomly assign 50\% of them to zero to introduce a moderately sparse setting. In addition, the nonzero parameters are then generated by 
sampling from a uniform distribution over the set $[-1, -0.4] \cup [0.4, 1]$. The intercept parameters $s_{jj},j=1,\ldots, J$ are set to zero. 
The true parameter values are given in 
the supplementary material. 
Missing data are simulated by masking particular elements under an MAR mechanism.  In particular, we have $z_{i6}=1$, so that the sixth variable is always observed. We further allow the missingness probabilities of the first five variables (i.e., $z_{ij}=0,j=1,\ldots,5$) to depend on the values of $y_{i6}$.
The specific settings on $p(z_{ij}=0\mid y_{i6}),j=1,\ldots, 5$ are detailed in the supplementary material.
Data are generated following the aforementioned Ising model and MAR mechanism for four different sample sizes, $N = 1,000, 2,000, 4,000$, and $8,000$, respectively. 
For each sample size,  50 independent replications are created.

Three methods are compared -- the proposed method, Bayesian inference with a single imputation, and Bayesian inference based on complete cases from listwise deletion. 
The Bayesian inference for complete data is performed by sampling parameters from the posterior implied by the pseudo-likelihood and a normal prior, which is a special case of the proposed method without iterative imputation steps.
All these methods shared the same initial values $s_{jl}^{(0)}\sim U(-0.1,0.1),1\leq j\leq l\leq J$.
For our proposed method, we set the length of the Markov Chain Monte Carlo (MCMC) iterations to $T = 5,000$ and a burn-in size of $T_0 = 1,000$, with a thinning parameter $k_0=10$. This setup leads to an effective total of 400 MCMC samples for the Ising parameter matrix $\bfS$. Notably,  identical MCMC length and burn-in configuration are applied  during parameters inference in the single imputation and complete-case analyses.

Figure~\ref{fig:3node} gives the plots for the mean squared errors (MSE) of the estimated edge parameters and intercept parameters under different sample sizes and for different methods. 
The MSE for each parameter $s_{jl}$ is defined as 
\begin{equation}
    \frac{1}{50} \sum_{k=1}^{50}(\hat{s}_{k,jl}-s_{0,jl})^2.
\end{equation}
Here, $\hat{s}_{k,jl}$ denotes the estimated value from the $k$th replication while $s_{0,jl}$ refers to the true value.
Each box in panel (a) corresponds to the 15 edge parameters, and each box in panel (b) corresponds to the 6 intercept parameters. 
We notice that the listwise deletion procedure introduces biases into the edge and intercept estimation, resulting in the MSEs for certain parameters not decaying toward zero as the sample size grows. Additionally, both the proposed method and the single imputation method offer accurate parameter estimation, with MSEs decaying toward zero as the sample size increases. Notably, the proposed method is substantially more accurate than the single imputation method, suggesting that aggregating over multiple imputed datasets improves the estimation accuracy. Furthermore, 
for smaller sample sizes, the complete-case analysis 
seems to yield slightly more accurate estimates of the edge parameters than the single imputation method. Across four sample sizes, the median computational times for obtaining the results of the proposed method were 33, 50, 88, and 185 seconds, respectively\footnote{All simulations were performed on Intel-based systems with the following configuration: Ubuntu 22.04.3 LTS operating system; Intel(R) Xeon(R) Platinum 8375C CPU @ 2.90GHz; Python version 3.11.2; Numpy version 1.23.5.}.

\begin{figure}[ht]
\centering
\begin{subfigure}{.5\textwidth}
  \centering
  \includegraphics[width=.99\linewidth]{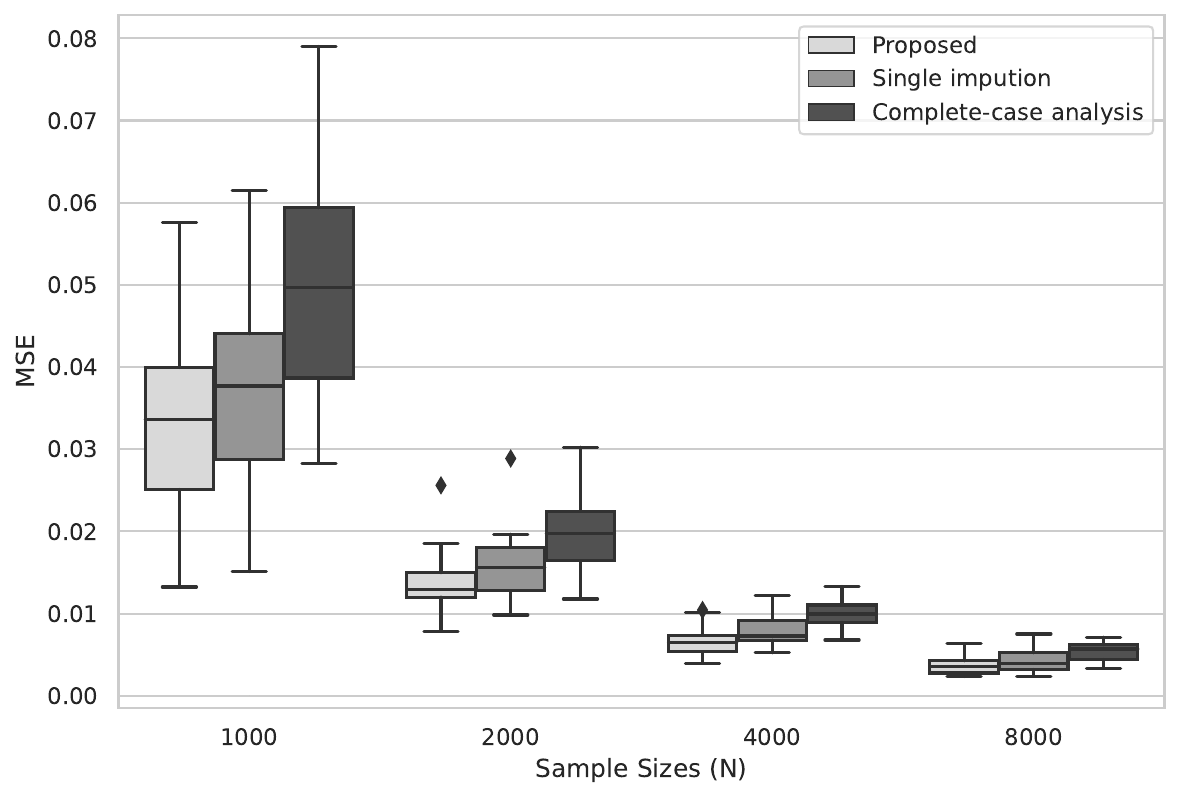}
  \caption{ }
  \label{fig:simu1_edges}
\end{subfigure}%
\begin{subfigure}{.5\textwidth}
  \centering
  \includegraphics[width=.99\linewidth]{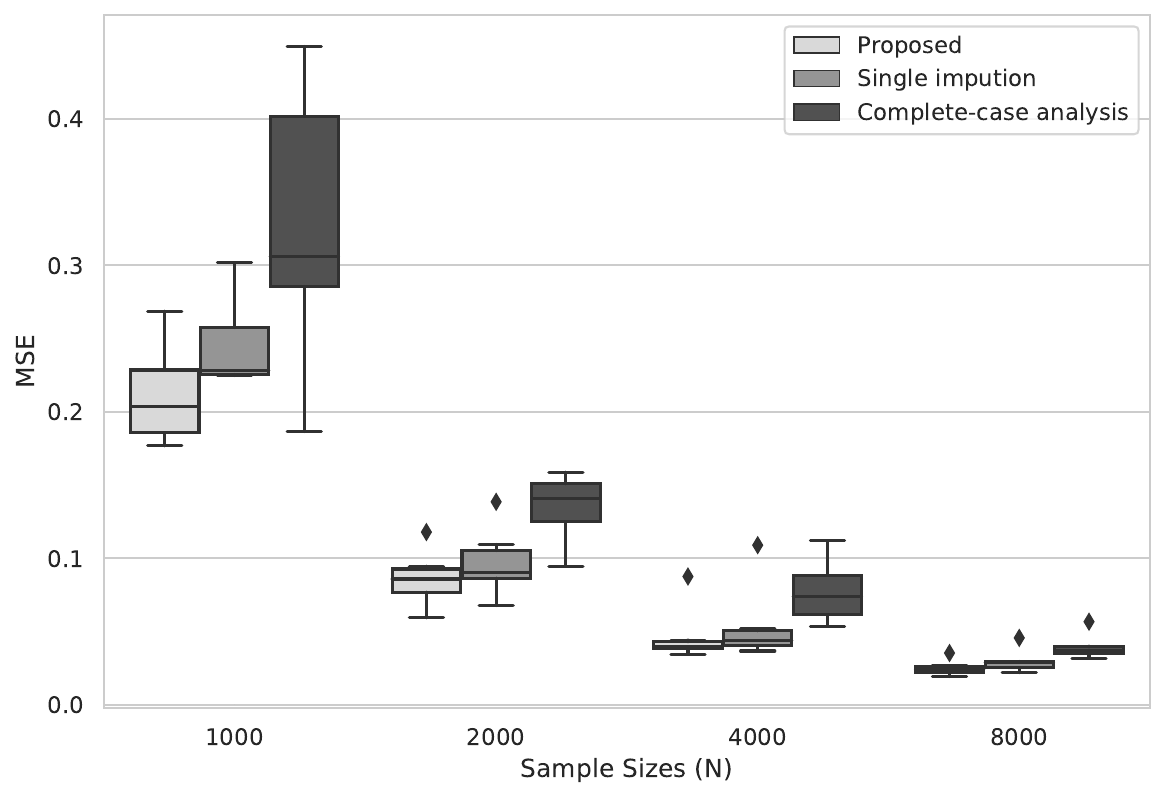}
  \caption{ }
  \label{fig:simu1_intercept}
\end{subfigure}
\caption{Panel (a): Boxplots of MSEs for edge parameters $s_{jl}$. Panel (b): Boxplots of MSEs for intercept parameters $s_{jj}$. }
\label{fig:3node}
\end{figure}

\paragraph*{Study II: Missing due to screening items}

Missingness due to screening items is commonly encountered in practice, posing challenges to the network analysis \citep[]{borsboom2017false,mcbride2023quantifying}. This occurs, for example, in surveys where initial screening questions determine respondents' eligibility or relevance for subsequent questions. 
Suppose respondents indicate a lack of relevant experience (i.e., their answers to the screening items are all negative). In that case, they are not prompted to answer certain follow-up questions, making the missingness of these responses depend on their answers to the screening questions and, thus, MAR. Our real data example in Section~\ref{sec:real_data_app} involves two screening items, which results in a large proportion of missing data. 

We consider a simulation setting involving two screening items to evaluate the proposed method's performance under this setting. 
Similar to Study I, we consider a setting with six items, the first two of which are the screening items. The full data are generated under an Ising model, whose parameters are given in the supplementary material, where the corresponding network has six positive edges including one between the two screening items.
The responses to the screening items are always set as observed for any observation. 
When an observation's responses to the screening items are both zero, their responses to the rest of the four items are regarded as missing. 

We consider a single sample size $N=8,000$ and generate 50 independent datasets. We apply the proposed method, the single imputation method, and the complete-case analysis. For each estimation procedure, we set the MCMC iterations $T = 5,000$,  the burn-in size $T_0 = 1,000$, and the thinning parameter $k_0=10$. These methods are compared in terms of MSEs and biases for parameter estimation. 
 
Table~\ref{tab:simu2} presents the result. For all the edge parameters except for $s_{12}$, the three estimation methods work well, though the single imputation method is slightly less accurate, as indicated by its slightly larger MSEs. However, the complete-case estimate is substantially negatively biased for $s_{12}$, the edge between two screening items. At the same time, the imputation-based methods are still accurate, with the proposed method having a smaller MSE than that of the single imputation method. This result confirms that running a complete-case analysis on data involving screening items is problematic while performing the imputation-based methods, especially the proposed method, yields valid results. 

We provide discussions on this result. The negative bias for $s_{12}$ in the complete-case analysis is due to a selection bias, typically referred to as Berkson's paradox \citep[]{de2021psychological}. The complete-case analysis excludes all the response vectors with negative responses to both screening items. Consequently, a positive response on one screening item strongly suggests a negative response on the other, regardless of the responses to the rest of the items. This results in a falsely negative conditional association between the two screening items. 
In fact, one can theoretically show that the frequentist estimate of $s_{12}$ based on the maximum pseudo-likelihood is negative infinity. The finite parameter estimate in Table~\ref{tab:simu2} for $s_{12}$ is due to the shrinkage effect of the prior distribution that we impose.  On the other hand, the proposed method makes use of the observed frequency of the (0,0) response pattern for the two screening items, in addition to the frequencies of the fully observed response vectors. As shown by the identifiability result in the Appendix, these frequencies are sufficient for identifying all the parameters of the Ising model.

\begin{table}[H]
\centering
\begin{tabular}{@{}cccc@{}}
\toprule
Edge & Proposed & Single Imputation & Complete-case Analysis \\ \midrule
& MSE $|$ Bias& MSE $|$ Bias& MSE $|$ Bias\\
\hline
$s_{12}$  & 0.007 $|$ \phantom{-}0.029 & 0.007 $|$ \phantom{-}0.032 & 57.060 $|$ -7.524\\
$s_{13}$  & 0.012 $|$ \phantom{-}0.030 & 0.019 $|$ \phantom{-}0.020 & 0.012 $|$ \phantom{-}0.030\\
$s_{14}$  & 0.010 $|$ \phantom{-}0.008 & 0.011 $|$ -0.002 & 0.010 $|$ \phantom{-}0.011\\
$s_{15}$  & 0.014 $|$ -0.007 & 0.020 $|$ -0.002 & 0.014 $|$ -0.005\\
$s_{16}$  & 0.011 $|$ -0.020 & 0.017 $|$ -0.027 & 0.011 $|$ -0.017\\
$s_{23}$  & 0.003 $|$ \phantom{-}0.012 & 0.004 $|$ \phantom{-}0.009 & 0.004 $|$ \phantom{-}0.012\\
$s_{24}$  & 0.005 $|$ \phantom{-}0.004 & 0.005 $|$ \phantom{-}0.001 & 0.005 $|$ \phantom{-}0.004\\
$s_{25}$  & 0.007 $|$ \phantom{-}0.011 & 0.009 $|$ \phantom{-}0.013 & 0.007 $|$ \phantom{-}0.012\\
$s_{26}$  & 0.005 $|$ -0.005 & 0.006 $|$ -0.008 & 0.005 $|$ -0.003\\
$s_{34}$  & 0.004 $|$ -0.004 & 0.004 $|$ -0.003 & 0.004 $|$ -0.003\\
$s_{35}$  & 0.006 $|$ \phantom{-}0.014 & 0.007 $|$ \phantom{-}0.017 & 0.006 $|$ \phantom{-}0.014\\
$s_{36}$  & 0.008 $|$ -0.009 & 0.007 $|$ -0.017 & 0.008 $|$ -0.009\\
$s_{45}$  & 0.006 $|$ -0.004 & 0.006 $|$ -0.003 & 0.006 $|$ -0.004\\
$s_{46}$  & 0.006 $|$ -0.002 & 0.006 $|$ -0.003 & 0.006 $|$ -0.003\\
$s_{56}$  & 0.007 $|$ \phantom{-}0.001 & 0.008 $|$ -0.002 & 0.007 $|$ \phantom{-}0.002\\
\bottomrule
\end{tabular}
\caption{MSEs and biases for edge parameters.}
\label{tab:simu2}
\end{table}

\subsection{A Real Data Application}\label{sec:real_data_app}

We analyze the dataset for the 2001-2002 National Epidemiological Survey of Alcohol and Related Conditions (NESARC), which offers valuable insights into alcohol consumption and associated issues in the U.S. population \citep[]{grant2003source}. The dataset consists of 43,093 civilian non-institutionalized individuals aged 18 and older. In this analysis, we focus on two specific sections of the survey 
that concern two highly prevalent mental health disorders -- Major Depressive Disorder (MDD) and Generalized Anxiety Disorder (GAD).  Because MDD and GAD have high symptom overlap \citep[]{hettema2008nosologic} 
and often co-occur \citep[]{hasin2005epidemiology},
it is important to perform a joint analysis of the symptoms of the two mental health disorders and study their separation. In particular, \cite{blanco2014latent} performed factor analysis based on the same data and found that the two mental health disorders have  distinct latent structures. We reanalyze the data, hoping to gain some insights from the network perspective of the two mental health disorders. 

Following \cite{blanco2014latent}, we consider data with nine items measuring MDD and six items measuring GAD. These items are designed according to the Diagnostic and Statistical Manual of Mental Disorders, Fourth Edition (DSM-IV) \citep[]{APA2000}.
These items ask the participants if they have recently
experienced certain symptoms; see  
Table~\ref{tab:real_data_items} for their short descriptions. 
After eliminating samples with entirely absent values across the 15 items, a total of 42,230 cases remain in the dataset. Note that any ``Unknown'' responses in the original data are converted into missing values. 
The dataset exhibits a significant degree of missingness, with only 2,412 complete cases for the 15 items, representing approximately 6\% of the total cases. Specifically, the missing rate for each item is given in Table~\ref{tab:real_data_items}. 
Importantly, items D1 and D2 function as screening items and, thus, have a very low missing rate. The respondents did not need to answer items D3-D9 if the responses to D1 and D2 were ``No'' or ``Unknown'', resulting in high missing rates for these items.  This pattern suggests that the missing data in this study is not MCAR. The GAD items A1-A6 also have a screening item, which results in the high missing rates in A1 through A6. Following the treatment in \cite{blanco2014latent}, these screening items are not included in the current analysis. 

\begin{table}[]
\centering
\begin{tabular}{ll}
\hline
\multicolumn{2}{c}{\textbf{MDD Item Description}} \\
\hline
D1 (0.1\%): Depressed mood & D5 (68.5\%): Psychomotor agitation/retardation \\
D2 (0.2\%): Diminished interest & D6 (68.0\%): Fatigue/loss of energy \\
D3 (68.5\%): Weight loss or gain & D7 (67.9\%): Feelings of guilt \\
D4 (67.9\%): Insomnia or hypersomnia & D8 (67.9\%): Diminished concentration \\
 & D9 (67.7\%): Recurrent thoughts of death \\
\hline
\multicolumn{2}{c}{\textbf{GAD Item Description}} \\
\hline
A1 (91.8\%): Restlessness & A4 (91.8\%): Irritability \\
A2 (91.9\%): Easily fatigued & A5 (91.9\%): Muscle tension \\
A3 (91.8\%): Difficulty concentrating & A6 (91.8\%): Sleep disturbance \\
\hline
\end{tabular}
\caption{Descriptions of MDD and GAD items and their missing rates.}
\label{tab:real_data_items}
\end{table}

We apply the proposed method and the complete-case analysis to the data. For each method,  10 MCMC chains with random starting values are used, each having 10,000 MCMC iterations and a burn-in size 5,000. The Gelman-Rubin statistics are always below 1.018, confirming the satisfactory convergence of all 120 parameters for both methods. The estimated network structures for MDD and GAD items are presented in Figure~\ref{fig:realdata}, 
where an edge is shown between two variables when the absolute value of the estimated parameter is greater than 0.5. We emphasize that this threshold is applied only for visualization purposes, rather than for edge selection. Consequently, the edges in Figure \ref{fig:realdata} should only be interpreted as edges with large estimated parameters, rather than truly nonzero edges. The nine MDD items are shown as blue nodes at the bottom, and the six GAD items are shown as orange nodes at the top. 
The edges are colored blue and orange, which represent
positive and negative parameter estimates, respectively.
In addition, the line thickness of the edges indicates their magnitude. A clear difference between the two methods is the edge between D1 ``depressed mood most of the day, nearly every day," and D2 ``markedly diminished interest or pleasure in all, or almost all, activities most of the day, nearly every day", which are two screening questions in the survey that all the participants responded to. The estimated parameter for this edge has a large absolute value under each of the two methods, but the estimated parameter is negative in the complete-case analysis, while it is positive according to the proposed method. As revealed by the result of Study II in Section~\ref{sec:6node}, the negative edge estimate of the edge between the screening items given by the complete case analysis is spurious. Considering the content of these items, we believe that the estimate from the proposed method is more sensible. Other than this edge, the remaining structure of the two networks tends to be similar, but with some differences. In particular, we see that the complete-case analysis yields more edges than the proposed method; for example, the edges of A4-A5, A1-D5, D1-D6, D1-D7, D1-D8, and D8-D9 appear in the estimated network from the complete-case analysis but not in that of the proposed method, while only two edges, A3-A5 and D3-D4, are present in the network estimated by the proposed method but absent in the network from the complete-case analysis. We believe this is due to the higher estimation variance of the complete-case analysis caused by its relatively small sample size.

Finally, our analysis shows that the symptoms of each mental health disorder tend to densely connect with each other in the Ising network, while the symptoms are only loosely but positively connected between the two mental health disorders. The edges between the two mental health disorders identify the overlapping symptoms, including  ``D4: Insomnia or hypersomnia"
and ``A6: Sleep disturbance", ``A2: Easily fatigued" and ``D6: Fatigue/loss of energy", and ``A3: Difficulty concentrating" and ``D8: Diminished concentration". These results suggest that MDD and GAD are two well-separated mental health disorders, despite their high symptom overlap and frequent co-occurrence. This result confirms the conclusion of \cite{blanco2014latent} that GAD and MDD are closely related but different nosological entities. 

\begin{figure}
\centering
\begin{subfigure}{.5\textwidth}
  \centering
  \includegraphics[width=.7\linewidth]{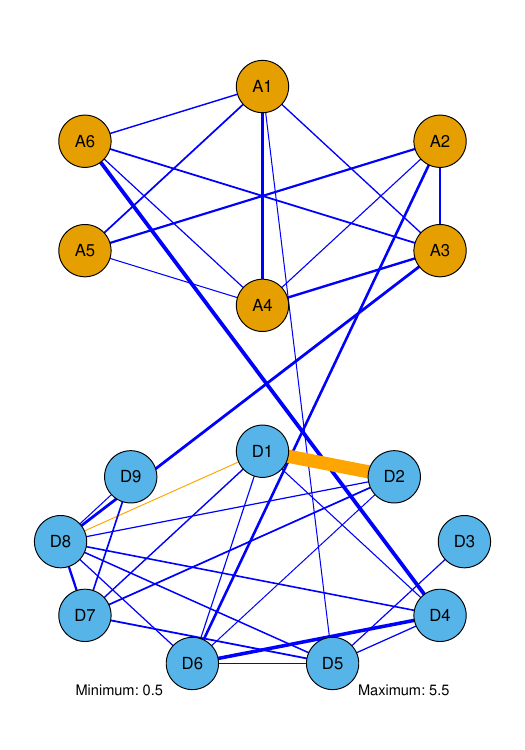}
  \caption{ }
  \label{fig:realdata_complete_analysis}
\end{subfigure}%
\begin{subfigure}{.5\textwidth}
  \centering
  \includegraphics[width=.7\linewidth]{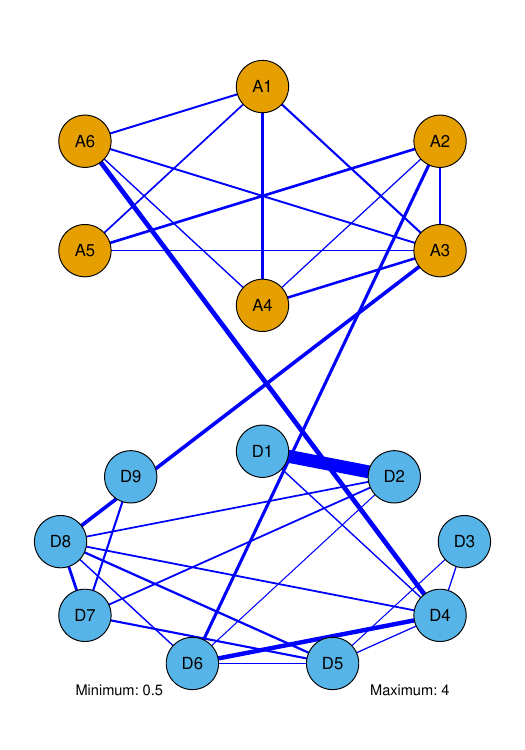}
  \caption{ }
  \label{fig:realdata_proposed}
\end{subfigure}
\caption{Estimated network structure for MDD and GAD. Panel (a): Complete-case analysis. Panel (b) Proposed method.}
\label{fig:realdata}
\end{figure}

\section{Concluding Remarks}\label{sec:conclusion}

In this paper, we propose a new method for Ising network analysis in the presence of missing data. The proposed method integrates iterative imputation into a Bayesian inference procedure based on conditional likelihood. An asymptotic theory is established that guarantees the consistency of the proposed estimator. Furthermore, a {P{\'o}lya}-Gamma machinery is proposed for the sampling of Ising model parameters, which yields efficient computation. 
The power of the proposed method is further shown via simulations and a real-data application. An R package has been developed that will be made publicly available upon the acceptance of the paper. 


The current work has several limitations that require future theoretical and methodological developments. 
First, this manuscript concentrates on parameter estimation for the Ising model in the presence of missing data. However, the problem of edge selection \citep{rovckova2018bayesian,noghrehchi2021selecting,borsboom2022possible,marsman2022objective} requires future investigation. There are several possible directions. One direction is to view it as a multiple-testing problem and develop procedures that control certain familywise error rates or the false discovery rate for the selection of edges. To do so, one needs to develop a way to quantify the uncertainty for the proposed estimator. It is nontrivial, as the proposed method is not a standard Bayesian procedure, and we still lack a theoretical understanding of the asymptotic distribution of the proposed procedure. In particular, it is unclear whether the Bernstein-von-Mises theorem that connects Bayesian and frequentist estimation holds under the current setting. Another direction is to view it as a model selection problem. In this direction, we can use sparsity-inducing priors to better explore the Ising network structure when it is sparse. We believe that the proposed method, including the iterative imputation and the  {P{\'o}lya}-Gamma machinery, can be adapted when we replace the normal prior with the spike-and-slab prior considered in \cite{marsman2022objective}. This adaptation can be done by adding some Gibbs sampling steps. In addition, it is of interest to develop an information criterion that is computationally efficient while statistically consistent. This may be achieved by computing an information criterion, such as the Bayesian information criterion, for each imputed dataset and then aggregating them across multiple imputations. Finally, the proposed method has several variants that may be useful for problems of different scales. For problems of a relatively small scale (i.e., when $J$ is small), we may perform data imputation using the sampled $\mathbf S$ instead of using auxiliary parameters $\bbb_j$s. This choice will make the algorithm computationally more intensive, as the sampling of $\mathbf S$ has a high computational complexity. On the other hand, it may make the estimator statistically more efficient as it avoids estimating the auxiliary parameters $\bbb_j$s, whose dimension is higher than $\mathbf S$. For very large-scale problems, one may estimate the Ising model parameters based only on the auxiliary parameters $\bbb_j$s. For example, we may estimate $s_{ij}$ by averaging the value of $(\beta_{ij}+\beta_{ji})/2$ over the iterations. This estimator is computationally more efficient than the proposed one, as it avoids sampling $\mathbf S$ given the imputed datasets. This estimator should still be consistent but may be statistically slightly less efficient than the proposed one.

\section*{Acknowledgement}\label{sec:acknoledgement}
The research was supported in part by the Shanghai Science and Technology Committee Rising-Star Program (22YF1411100).

\clearpage

\setcounter{equation}{0}
\renewcommand{\theequation}{S\arabic{equation}}
\setcounter{theorem}{0}
\renewcommand{\thetheorem}{S\arabic{theorem}}
\setcounter{definition}{0}
\renewcommand{\thedefinition}{S\arabic{definition}}

\appendix
\section*{Appendix}
\section{Identifiability of Ising Model with Two Screening Items}

We investigate the identifiability of the Ising model parameters when there are two screening items.  
\begin{proposition}
Consider an Ising model with $J \geq 3$, true parameters $\bfS_0 = (s_{ij}^0)_{J\times J}$, and the first two items being the screening items. We define $p_0 (\mathbf y) := P(\mathbf Y = \mathbf y|\bfS_0), $ for any $\mathbf y \in \mathcal A = \{(x_1, ..., x_J)^\top \in \{0,1\}^J:x_1 = 1 \mbox{~or~} x_2=1 \}$, and $p_0(0,0) := P(Y_1 = 0, Y_2 =0|\bfS_0)$ under the Ising model. Then there does not exist an Ising parameter matrix $\bfS \neq \bfS_0$ such that  $p_0 (\mathbf y) = P(\mathbf Y = \mathbf y|\bfS), $ for any $\mathbf y \in \mathcal A$ and $p_0(0,0) = P(Y_1 = 0, Y_2 =0|\bfS)$ under the Ising model.

\end{proposition}
\begin{proof}

We first prove the statement for $J\geq 4$. Suppose that $p_0 (\mathbf y) = P(\mathbf Y = \mathbf y|\bfS), $ for any $\mathbf y \in \mathcal A$ and $p_0(0,0) = P(Y_1 = 0, Y_2 =0|\bfS)$. We will prove that $\bfS =\bfS_0$.

We start by considering items $1, 2, 3, 4$. We define the set $\mathcal A_{3,4} =   \{(x_1, ..., x_J)^\top \in \mathcal A: x_5 = ... = x_J=0\}$. We note that $\mathcal A_{3,4}$ has 12 elements. Using $\bfy_d = (1, 0,0,0,0,...,0)^\top \in \mathcal A_{3,4}$ as the baseline pattern,  for any $\bfy \in \mathcal A_{3,4}$ such that $\bfy \neq \bfy_d$ we have 

$$ \log\left[\frac{\exp(\frac{1}{2}\bfy^\top\bfS\bfy)}{c(\bfS)} / \frac{\exp(\frac{1}{2}\bfy_d^\top\bfS\bfy_d)}{c(\bfS)}\right] = \frac{1}{2}\bfy^\top\bfS\bfy-\frac{1}{2}\bfy_d^\top\bfS\bfy_d = \frac{1}{2}\bfy^\top\bfS_0\bfy-\frac{1}{2}\bfy_d^\top\bfS_0\bfy_d.$$
That gives us 11 linear equations for 10 parameters, $s_{ij}$, $i,j\leq 4$. By simplifying these equations, we have 
1) two linear equations for $(s_{11}, s_{12}, s_{22})$
\begin{equation}\label{eq:leqs1}
    \begin{aligned}
        s_{11} - s_{22} &= s_{11}^0 - s_{22}^0\\
        2s_{12} + s_{22} &=  2s_{12}^0 + s_{22}^0
    \end{aligned}
\end{equation}
and 2) $s_{1i} = s_{1i}^0$, $s_{2i} = s_{2i}^0$, $s_{ij} = s_{ij}^0$ for all $i,j = 3, 4$. 

We repeat the above calculation for any four-item set involving items 1 and 2. By choosing any item pair $i,j > 2$, $i\neq j$ and repeating the above calculation with patterns in the set 
$\mathcal A_{i,j} =   \{(x_1, ..., x_J)^\top \in \mathcal A: x_l = 0, l \neq 1, 2, i, j\}$, we have  $s_{1i} = s_{1i}^0$, $s_{2i} = s_{2i}^0$, 
$s_{ij} = s_{ij}^0$ for all $i,j > 2$, $i\neq j$. With the above argument and \eqref{eq:leqs1}, we only need to show $s_{11} = s_{11}^0$ to prove $\bfS = \bfS_0$. To prove $s_{11} = s_{11}^0$, we use $p_0(0,0)/p_0(\bfy_d) = P(Y_1 = 0, Y_2 =0|\bfS)/P(\bfY = \bfy_d|\bfS)$ and that $s_{ij} = s_{ij}^0$ for all $i,j > 2$, $i\neq j$ which we have proved. We have 

$$\sum_{\bfy \in \mathcal A_0} \exp\left(-\frac{1}{2}s_{11}^0 + \frac{1}{2}\bfy^\top\bfS_0\bfy\right) = \sum_{\bfy \in \mathcal A_0} \exp\left(-\frac{1}{2}s_{11} + \frac{1}{2}\bfy^\top\bfS_0\bfy\right),$$
where $\mathcal A_0 = \{(x_1, ..., x_J)^\top \in \{0,1\}^J: x_1 = x_2 = 0\}$. As the right-hand side of the above equation is a strictly monotone decreasing function of $s_{11}$, we know that $s_{11} = s_{11}^0$ is the only solution to the above equation. This proves the $J \geq 4$ case. 

We now move on to the case when $J = 3$. We consider $\mathcal A = \{(x_1, x_2, x_3)^\top \in \{0,1\}^3:x_1 = 1 \mbox{~or~} x_2=1 \}$ and $\bfy_d = (1, 0,0)^\top$. Using $\bfy_d$ as the baseline, for any $\bfy \in \mathcal A$, $\bfy \neq \bfy_d$, we construct five linear equations given by $\log(p_0(\bfy)/p_0(\bfy_d) )= \log(P(\bfY =\bfy\vert\bfS)/P(\bfY =\bfy_d\vert\bfS))$. From these equations, we obtain: 
1) two linear equations for $(s_{11},s_{12},s_{22})$
\begin{equation}
    \begin{aligned}
        s_{11} - s_{22} &= s_{11}^0 - s_{22}^0 \\
        2s_{12} + s_{22} &= 2s_{12}^0 + s_{22}^0
    \end{aligned}
\end{equation}
and 2) $s_{13} = s_{13}^0, s_{23}=s_{23}^0, s_{33} = s_{33}^0$. Again, with the above equations, $\bfS = \bfS^0$ if $s_{11} = s_{11}^0$. To show $s_{11} = s_{11}^0$, we use $p_0(0,0)/p_0(\bfy_d) = P(Y_1 = 0, Y_2 =0|\bfS)/P(\bfY = \bfy_d|\bfS)$ and $s_{33} = s_{33}^0$. We have 
$$\exp(-\frac{1}{2}s_{11}^0) + \exp(-\frac{1}{2}s_{11}^0 + \frac{1}{2}s_{33}^0) =  \exp(-\frac{1}{2}s_{11}) + \exp(-\frac{1}{2}s_{11} + \frac{1}{2}s_{33}^0).$$
As the right-hand side is a strictly monotone decreasing function of $s_{11}$, we know that $s_{11} = s_{11}^0$ is the only solution to the above equation. This proves the $J=3$ case, which completes the proof. 
\end{proof}

Following the same proof strategy as above, it can be further shown that  
$s_{11}, s_{12},$ and $s_{22}$ are not identified in the complete-case analysis, while the rest of the parameters are. This is consistent with the result of Simulation Study II, where the other model parameters can be accurately estimated while the estimates of $s_{11}, s_{12},$ and $s_{22}$ are substantially different from the corresponding true parameters. 

\section{Technical proofs}\label{sec:technical_proofs}

\subsection{A lemma for imputation consistency}\label{subsec:verify_liu_theorem1}
Following the derivation of Section 2.3 in the main text, under ignorable missingness assumption, the posterior distribution for $\bfS$ satisfies
$\pi_N(\bfS) \propto p(\bfy_{obs}\mid\bfS) \pi(\bfS)$. Under the same Bayesian model, one can impute the missing values from the posterior predictive distribution. That is, the posterior predictive distribution for $\bfy_{i,mis}$, $i=1, ..., N$,   takes the form 
$$p_N(\bfy_{1,mis}, ..., \bfy_{N,mis}) = \int \pi_N(\bfS)\left(\prod_{i=1}^N \frac{\exp\left(\bfy_i^\top\bfS\bfy_i/2\right)}{(c(\bfS)p_i(\bfy_{i,obs}\mid \bfS))} \right) d\bfS,$$ 
where $p_i(\bfy_{i,obs}\mid \bfS) = \sum_{y_{ij}:z_{ij}=0} \exp\left(\bfy_i^\top\bfS\bfy_i/2\right)/c(\bfS)$. Further, suppose that the 
Algorithm~\ref{alg:pseudo_imputation} converges to a stationary distribution, and let $p^*_N(\bfy_{1,mis}, ..., \bfy_{N,mis})$ be the implied posterior predictive distribution given the observed data. 
Then we show in Lemma~\ref{lemma:data_consist}, which is an adaptation of Theorem 1 of \cite{liu2014stationary}, suggests that $p_N(\bfy_{1,mis}, ..., \bfy_{N,mis})$ and $p^*_N(\bfy_{1,mis}, ..., \bfy_{N,mis})$ converge to each other in the total variation sense. 

\begin{lemma}\label{lemma:data_consist}
  Assume the following assumptions hold: 1) The Markov chain for missing data, generated by the iterative imputation algorithm Algorithm~\ref{alg:pseudo_imputation}, is positive Harris recurrent and thus admits a unique stationary distribution denoted by $p_N^*$; 2) The missing data process is ignorable; 3) A regularity condition holds for prior distributions of Ising model parameters and auxiliary parameters, as detailed in \Cref{assume:assumption_for_priro}. 
  Then the implied posterior predictive distribution $p^*_N$ is consistent with the true posterior predictive distribution, $p_N$, i.e.,
  \begin{equation}
  \begin{aligned}
      d_{TV}\left(p^*_N,p_N\right) = \max_{\bfy_{1,mis},\ldots,\bfy_{N,mis}} \vert p^*_N(\bfy_{1,mis}, ..., \bfy_{N,mis}) - p_N(\bfy_{1,mis}, ..., \bfy_{N,mis})\vert \rightarrow 0,
  \end{aligned}
  \end{equation}
  in probability as $N\rightarrow\infty$. 
  
\end{lemma}

\begin{algorithm}\caption{Ising model Gibbs chain}\label{alg:joint_imputation}
  \KwData{Given observed data, initial values for the Ising model parameters $\bfS^{(0)}$, randomly imputed missing data $\bfy_1^{(0)},\ldots,\bfy_N^{(0)}$, MCMC length $T$, burn-in size $T_0$.}
  \For{iteration $t$ in range $1$ to $T$}{
    \For{each $j = 1$ to $J$}{      
  
      {Sample $\bfS^{(t)}$ from $p(\bfS \mid \bfy_1,\ldots,\bfy_N)\propto\pi(\bfS)\prod_{i=1}^N\exp(\bfy_i^\top\bfS\bfy_i/2)/c(\bfS).$}

      {Sample $y_{ij}^{(t)}$ for each $i\in\Omega_j$ from the Bernoulli distribution $p(y_{ij}\mid\bfy_{i,-j},\bfS^{(t)})$ and update $\bfy_1,\ldots,\bfy_N$ with the sampled values.} 
    }
  
    } 
  \KwResult{$\hat\bfS = \frac{1}{T-T_0}\sum_{t=T_0+1}^T\bfS^{(t)}$.}
\end{algorithm}

To prove Lemma~\ref{lemma:data_consist}, we start by define a Gibbs sampling process for the joint Ising model, as outlined in Algorithm~\ref{alg:joint_imputation}. This algorithm is constructed for the theoretical purposes since the step of sampling $\bfS$ is intractable. 
The aim of our proof is to validate the posterior predictive distribution of missing data given observed data $p_N^*$, implied by the Algorithm~\ref{alg:pseudo_imputation}, converges in total variation to the true posterior predictive distribution $p_N$.
We first establish that $p_N^*$ in fact converges to the stationary distribution of the Gibbs chain, denoted as $p_N'$, implied by Algorithm~\ref{alg:joint_imputation}. By corroborating the convergence of $p_N'$ and $p_N$, the proof of Lemma~\ref{lemma:data_consist} is thereby completed. In the following proof, we reparameterize $\bfalpha=\text{vech}(\bfS)$ for convenience. We define the following for the proof. 
\begin{definition}
  \quad 
  \begin{itemize}
    \item We denotes $\mathcal{Y}$ the data matrix with $N$ samples and $J$ variables, $\mathcal{Y}_j$ the $j$th column and $\mathcal{Y}_{-j}$ the remaining $j-1$ columns.
    \item Define $A_N=\{\mathcal{Y}\mid \Vert\hat\bfalpha(\mathcal{Y})\Vert\leq\gamma\},$ where $\hat\bfalpha(\mathcal{Y})$ is the complete-data maximum likelihood estimator, where $\gamma$ can be sufficiently large so that 
    \begin{equation}
      \begin{aligned}
        p_N^*(A_N) &\rightarrow 1,\text{ and}\\
        p_N'(A_N) &\rightarrow 1,
      \end{aligned}
    \end{equation}
    in probability as $N\rightarrow\infty$.
    \item Let  
    \begin{equation}
      K(\omega,d\omega') = \pr(\mathcal{Y}_{mis}^{(k+1)}\in d\omega'\mid \mathcal{Y}_{mis}^{(k)} = \omega)
    \end{equation}
    be the transition kernels for the missing data chain, which depend on $\mathcal{Y}_{obs}$.
    \item Let $K^*(\omega,d\omega')$ and $K'(\omega,d\omega')$ be the transition kernels for the missing data chains from Algorithm~\ref{alg:pseudo_imputation} and Algorithm~\ref{alg:joint_imputation}, respectively.
    \item We further define the transition kernels conditional on $A_N$ by 
    \begin{equation}
      \tilde K(\omega,B) = \frac{K(\omega,B\cap  A_N)}{K(\omega,A_N)}.
    \end{equation}
    So we have $\tilde K^*(\omega,\cdot),\tilde K'(\omega,\cdot)$ are two transition kernels for the missing data chains conditional on $A_N$. And let $\tilde{p}_N^*,$ $\tilde{p}_N'$ be their stationary distributions, respectively.
    \item Define $\Vert\mu\Vert_1 = \sup_{\vert h\vert\leq 1}\int h(x)\mu(dx)$.
  \end{itemize}
\end{definition}

\begin{assumption}[A regularity condition for priors]\label{assume:assumption_for_priro}

    Let $\hat\bfalpha(\mathcal{Y})$ be the complete data maximum likelihood estimator and $A_N = \{\mathcal{Y}: \Vert\hat\bfalpha(\mathcal{Y})\Vert\leq\gamma\}$. Since the logistic models are with Ising model, we also have map $\bbb_j=T_j(\bfalpha),j=1,\ldots,J.$ Let $\pi_j(\bbb_j)$ and $\pi(\bfalpha)$ be prior distributions. Further define $\bbb_j^* = T_j^*(\bfalpha)$ such that $\tilde{T}_j(\bfalpha) = \{T_j(\bfalpha),T_j^*(\bfalpha)\}$ is a one-to-one invertible map ($\bbb_j^*$ can be $\bfalpha\setminus\bbb_j$). 
    Define $$\pi_j^*(\bbb_j,\bbb_j^*) = \det(\partial\tilde{T}_j/\partial\bfalpha)^{-1}\pi(\tilde{T}_j^{-1}(\bbb_j,\bbb_j^*)).$$
    Let $L_j(\bbb_j) = \pi_j(\bbb_j)/\pi_{j,\mathcal{Y}_{-j}}(\bbb_j),$ where 
    \begin{equation}
        \begin{aligned}
            \pi_{j,\mathcal{Y}_{-j}}(\bbb_j) &= \int p(\mathcal{Y}_{-j}\mid \bbb_j,\bbb_j^*)\pi_j^*(\bbb_j,\bbb_j^*)d\bbb_j^*\\
            &=\int\sum_{y_{1j},\ldots,y_{Nj}} p(\mathcal{Y}_j,\mathcal{Y}_{-j}\mid \bbb_j,\bbb_j^*)\pi_j^*(\bbb_j,\bbb_j^*)d\bbb_j^*.
        \end{aligned}
    \end{equation}
    The assumption requires that on the set $A_N$, $$\sup_{\Vert\bbb_j\Vert<\gamma}\partial\log L_j(\bbb_j)<\infty.$$
\end{assumption}
We remark that the above assumption holds for the Ising model with the normal priors adopted in the current paper. Specifically, $\pi_j(\bbb_j)$ and $\pi(\bfalpha)$ are $J$-variate and $J(J+1)/2$-variate normal distributions, respectively. Moreover, $\pi_j^*$ can also be a $J(J+1)/2$ normal distribution. 
Since on $A_N$, $L_j(\bbb_j)$ is a continuously differentiable function defined in $\mathbb{R}^{J(J+1)/2}$, we have then it is Lipschitz on any compact set in $\mathbb{R}^{J(J+1)/2}$. That is, on $A_N$, $\sup_{\Vert\bbb_j\Vert<\gamma}\partial\log L_j(\bbb_j) = \sup_{\Vert\bbb_j\Vert<\gamma}\left[\partial \log\pi_j(\bbb_j) - \partial\log\pi_{j,\mathcal{Y}_{-j}}(\bbb_j)\right]<\infty.$

\begin{proof}[Proof of Lemma~\ref{lemma:data_consist}]
  According to the assumptions, the Markov chain for the missing data produced by the Gibbs sampling procedure Algorithm~\ref{alg:joint_imputation} is positive Harris recurrent and thus admit a unique stationary distribution $p_N'$.
  We verify the conditions  holds. First, on $A_N$, the Fisher information of the Ising model has a lower bound of $\epsilon n$ for some $\epsilon$. So according to proposition 1 of \cite{liu2014stationary}, we have $\Vert K^*(\omega,\cdot)-K'(\omega,\cdot)\Vert_1\rightarrow 0$ uniformly for $\omega\in A_N$, that is,
  \begin{equation}
    \lim_{N\rightarrow\infty}\Vert \tilde K^*(\omega,\cdot) - \tilde K'(\omega,\cdot)\Vert_1 = 0.
  \end{equation}
  According to the standard bound for Markov chain convergence rates, there exists a common starting value $\omega\in C$ and a bound $r_k$ such that (ii) of Lemma 2 holds. Then Lemma 2 gives us
  \begin{equation}
    d_{TV}(\tilde{p}_N^*, \tilde{p}_N')\rightarrow 0,
  \end{equation}
  Further combining with conclusions in Lemma 1 in \cite{liu2014stationary} that $d_{TV}(p_N^*, \tilde p_N^*)\rightarrow 0,$ and $d_{TV}(p_N', \tilde p_N')\rightarrow 0$, we have the convergence of iterative imputation of compatible models, 
  \begin{equation}\label{eq:consist_iter}
    d_{TV}(p_N^*, p_N') \rightarrow 0,
  \end{equation}
  in probability as $N\rightarrow \infty$.
  Next, based on the construction of the Gibbs sampling procedure Algorithm~\ref{alg:joint_imputation}, we have the sequence converges to the target distribution, that is,
  \begin{equation}\label{eq:consist_gibbs}
    d_{TV}(p_N', p_N) \rightarrow 0,
  \end{equation}
  in probability as $N\rightarrow \infty$.
  Based on \eqref{eq:consist_iter} and \eqref{eq:consist_gibbs}, we have
  \begin{equation}
    \begin{aligned}
      d_{TV}(p_N^*, p_N) &= \sup_{A\in\mathcal{F}} \vert p_N^*(A) - p_N(A)\vert\\
      & \leq \sup_{A\in\mathcal{F}} \vert p_N^*(A) - p_N'(A)\vert + \sup_{A\in\mathcal{F}} \vert p_N'(A) - p_N(A)\vert \rightarrow 0,
    \end{aligned}
  \end{equation} 
  in probability as $N\rightarrow \infty$.
\end{proof}

\begin{remark}
  \Cref{lemma:data_consist} emphasizes the consistency of the proposed iterative imputation process. It implies that the implied posterior predictive distribution gradually converges to the posterior predictive distribution under standard Bayesian inference, underscoring the validity of the iterative imputation.
\end{remark}

\subsection[Proof of theorem 1]{Proof of Theorem~\ref{thm:params_consist}}

We will use $\bfalpha$, the half-vectorization of $\bfS$ in the proof. Denote $\pi_N^*(\bfalpha)$ the posterior density of $\bfalpha$ implied by the stationary distribution of the proposed method. Let $\mathcal{Y}$ be the data matrix with $\mathcal{Y}_{mis}$ and $\mathcal{Y}_{obs}$ being the missing and observed parts, respectively. We have
\begin{equation}\label{eq:theorem2_conclusion}
\begin{aligned}
    &\int_{B_\varepsilon(\bfalpha_0)}\pi_N^*(\bfalpha)d\bfalpha \\
    =&\int_{B_\varepsilon(\bfalpha_0)}\left[ \sum_{\mathcal{Y}_{mis}} p^*(\mathcal{Y}_{mis},\mathcal{Y}_{obs}\mid\bfalpha)p^*(\mathcal{Y}_{mis}\mid\mathcal{Y}_{obs}) \right]\pi(\bfalpha) / c_N d\bfalpha\\
    =&\sum_{\mathcal{Y}_{mis}}\left[\int_{B_\varepsilon(\bfalpha_0)} p^*(\mathcal{Y}_{mis},\mathcal{Y}_{obs}\mid\bfalpha) \pi(\bfalpha) / c_N d\bfalpha\right] p^*(\mathcal{Y}_{mis}\mid\mathcal{Y}_{obs})\\
    =&\sum_{\mathcal{Y}_{mis}}\left[\int_{B_\varepsilon(\bfalpha_0)} \exp(-N f_N(\bfalpha)) \pi(\bfalpha) / c_N d\bfalpha\right] p^*(\mathcal{Y}_{mis}\mid\mathcal{Y}_{obs}),
\end{aligned}
\end{equation}
where $f_N(\bfalpha) = -\frac{1}{N}\sum_{i=1}^N\log p^*(\bfy_i\mid\bfalpha)$ given in \eqref{eq:ising_model_pseudo}, $c_N = \int \exp(-N f_N(\bfalpha)) \pi(\bfalpha)d\bfalpha$. 
We further let 
\begin{equation}\label{eq:pi_N}
\begin{aligned}
    &\int_{B_\varepsilon(\bfalpha_0)}\pi_N(\bfalpha)d\bfalpha \\
    =&\sum_{\mathcal{Y}_{mis}}\left[\int_{B_\varepsilon(\bfalpha_0)} \exp(-N f_N(\bfalpha)) \pi(\bfalpha) / c_N d\bfalpha\right] p(\mathcal{Y}_{mis}\mid\mathcal{Y}_{obs}).
\end{aligned}
\end{equation}

Let $\Theta\in \mathbb{R}^{J(J+1)/2}, E\subset\Theta$ be open and bounded. It can be veried that: 1) $f_N$ have continuous third derivatives; 2) $f_N\rightarrow f$ pointwise for some $f$; 3) $f''(\bfalpha_0)$ is positive definite; 4) $f'''(\bfalpha_0)$ is uniformly bounded on $E$; 5) each $f_N$ is convex and $f'(\bfalpha_0)=0$. 
Then, according to the generalized posterior concentration theorem \citep[see Theorem 5,][]{miller2021asymptotic}, we have for any $\varepsilon > 0$,
\begin{equation}
    \int_{B_\varepsilon(\bfalpha_0)}\exp(-N f_N(\bfalpha)) \pi(\bfalpha) / c_N d\bfalpha \rightarrow 1
\end{equation}
in probability as $N\rightarrow\infty.$ Consequently,
\begin{equation}\label{eq:theorem2_pi_N_converge}
    \int_{B_\varepsilon(\bfalpha_0)}\pi_N(\bfalpha)d\bfalpha \rightarrow 1,
\end{equation}
in probability as $N\rightarrow\infty.$
Finally, by employing the convergence of imputation from Lemma~\ref{lemma:data_consist}, specifically $$d_{TV}(p^*(\mathcal{Y}_{mis}\mid\mathcal{Y}_{obs}) - p(\mathcal{Y}_{mis}\mid\mathcal{Y}_{obs}))\rightarrow 0$$ in probability as $N\rightarrow\infty$, we arrive at 
\begin{equation}\label{eq:theorem2_diff_converge}
    \sum_{\mathcal{Y}_{mis}}\left[\int_{B_\varepsilon(\bfalpha_0)} \exp(-N f_N(\bfalpha)) \pi(\bfalpha) / c_N d\bfalpha\right] (p^*(\mathcal{Y}_{mis}\mid\mathcal{Y}_{obs}) - p(\mathcal{Y}_{mis}\mid\mathcal{Y}_{obs}))\rightarrow 0
\end{equation}
in probability as $N\rightarrow\infty.$ This conclude the proof, given that
\begin{equation}
    \int_{B_\varepsilon(\bfalpha_0)}\pi_N^*(\bfalpha)d\bfalpha = \int_{B_\varepsilon(\bfalpha_0)}\pi_N(\bfalpha)d\bfalpha + \int_{B_\varepsilon(\bfalpha_0)}(\pi_N^*(\bfalpha)-\pi_N(\bfalpha))d\bfalpha,
\end{equation}
where the first term converges to 1 (i.e., \eqref{eq:theorem2_pi_N_converge}) and the second term converges to 0 (i.e., \eqref{eq:theorem2_diff_converge}) in probability as $N\rightarrow\infty$.


\section[Computation details for sampling details]{Computation details for sampling $\bbb_j$ and $\bfS$}

Upon observation of the logistic form presented in the conditional distribution when sampling the auxiliary parameters $\bbb_j$, we employ the {P{\'o}lya}-Gamma for effective sampling. Denote a random variable $\omega$ follows the {P{\'o}lya}-Gamma distribution $PG(b,c), b>0, c\in\mathbb{R}$ with parameters $b>0$ and $c\in\mathbb{R}$ if it is a weighted sum of independent Gamma random variables
\begin{equation}
    \omega = \frac{1}{2\pi^2}\sum_{k=1}^\infty \frac{g_k}{(k-1/2)^2+c^2/(4\pi^2)},
\end{equation}
where $g_k\sim\Gamma(b,1)$, which is the Gamma distribution with shape and rate parameters as $b$ and $1$, respectively. 


\subsection[derivation of posterior sj]{Derivation of posterior distribution of $\bbb_j$}\label{sec:derive_beta_j}
By introducing {P{\'o}lya}-Gamma latent variables $\omega_{ij}\sim PG(1,0),i=1,\ldots,N$, we establish a connection between the logistic form and the normal distribution. 
We rephrase the $j$th conditional distribution $p(y_{ij}\mid\bfy_{i,-j},\bbb_j)$ by the following equation,
\begin{equation}\label{eq:polya_gamma1}
  \frac{\exp(\phi_{ij})^{y_{ij}}}{1 + \exp(\phi_{ij})}
  = 2^{-1}\exp(\kappa_{ij}\phi_{ij})\mathbb{E}_{\omega_{ij}}\left[\exp(-\omega_{ij}\phi_{ij}^2/2)\right],
\end{equation} 
where $\kappa_{ij} = y_{ij} - 1/2$, $\omega_{ij}\sim PG(1,0)$, $\omega_{ij}\mid \phi_{ij}\sim PG(1, \phi_{ij}), \phi_{ij}=\beta_{jj}/2+\sum_{k\neq j}\beta_{jk}y_{ik}.$ 

Denote $\mathcal{Y}=(\bfy_1,\ldots,\bfy_N)^\top$, $\mathcal{Y}_j$ as the $j$th column of $\mathcal{Y}$, and $\mathcal{Y}_{-j}$ the remaining $j-1$ columns. Given $\bbb_j$, sample $N$ augmentation variables $\omega_{ij},i=1,\ldots,N$, each from a {P{\'o}lya}-Gamma distribution
\begin{equation}
    \omega_{ij}\mid\bbb_j,\bfy_i\sim PG(1,\beta_{jj}/2+\sum_{k\neq j}\beta_{jk}y_{ik}),
\end{equation}
based on \cref{eq:polya_gamma1}. Moreover, for the $j$th variable, we have
\begin{equation}\label{eq:augmented_condi}
  \begin{aligned}
    p(\mathcal{Y}_j\mid\mathcal{Y}_{-j},\bbb_j,\boldsymbol{\omega}_j) &= \prod_{i=1}^N p(y_{ij}\mid \bfy_{i,-j},\bbb_j,\omega_{ij})\\
    &= \prod_{i=1}^N 2^{-1}\exp(\kappa_{ij}(\beta_{jj}/2+\sum_{k\neq j}\beta_{jk}y_{ik}))\exp(-\omega_{ij}(\beta_{jj}/2+\sum_{k\neq j}\beta_{jk}y_{ik})^2/2)\\
    &\propto \exp\left[-\frac{1}{2}\left(\bbb_j^\top(\mathcal{Y} - \bfkappa_j\boldsymbol{e}_j^\top)^\top D_{\omega_j}(\mathcal{Y} - \bfkappa_j\boldsymbol{e}_j^\top)\bbb_j-2\bfkappa_j^\top(\mathcal{Y} - \bfkappa_j\boldsymbol{e}_j^\top)\bbb_j\right) \right],
  \end{aligned}
\end{equation}
where $\bfkappa_j=(\kappa_{1j},\ldots,\kappa_{Nj})^\top, \kappa_{ij} = y_{ij} - 1/2, D_{\omega_j} = \text{diag}(\boldsymbol{\omega}_j)$. We further have the following conditional distribution for $\bbb_j$
\begin{equation}
  \begin{aligned}
    p(\bbb_j&\mid\mathcal{Y},\boldsymbol{\omega}_j)\propto p(\mathcal{Y}_j\mid\mathcal{Y}_{-j},\bbb_j,\boldsymbol{\omega}_j)\pi_j(\bbb_j)\\
    &\propto \exp\left[-\frac{1}{2}\left(\bbb_j^\top(\mathcal{Y} - \bfkappa_j\boldsymbol{e}_j^\top)^\top D_{\omega_j}(\mathcal{Y} - \bfkappa_j\boldsymbol{e}_j^\top)\bbb_j-2\bfkappa_j^\top(\mathcal{Y} - \bfkappa_j\boldsymbol{e}_j^\top)\bbb_j\right)-\frac{1}{2}\bbb_j^\top D_{\beta_j}\bbb_j\right]\\
    &=\exp\left[-\frac{1}{2}(\bbb_j-\bfmu_{\beta_j})^\top\Sigma_{\beta_j}^{-1}(\bbb_j-\bfmu_{\beta_j})\right],
  \end{aligned}
\end{equation}
where $\Sigma_{\beta_j} = \left[(\mathcal{Y} - \bfkappa_j\boldsymbol{e}_j^\top)^\top D_{\omega_j}(\mathcal{Y} - \bfkappa_j\boldsymbol{e}_j^\top)+D_{\beta_j}\right]^{-1},$ $\bfmu_{\beta_j} = \Sigma_{\beta_j}(\mathcal{Y} - \bfkappa_j\boldsymbol{e}_j^\top)^\top\bfkappa_j$. 
Here, $\boldsymbol{e}_j$ is a $J$-dimensional vector with the $j$th element be one and all others be zeros, $D_{\omega_j}=\text{diag}(\boldsymbol{\omega}_j),$ $D_{\beta_j} = \text{diag}(\boldsymbol{\tau}_j),$ where $\tau_{jl}=\sigma_1^{-2},$ for $l\neq j$ and $\tau_{jj}=\sigma_2^{-2}.$ A weak informative prior on the intercept parameter by letting $\sigma_2^2>\sigma_1^2$.

To summarize, the introduced {P{\'o}lya}-Gamma latent variables $\omega_{ij}$ establish a connection between the logistic form and the normal distribution that lead to a normal form of the posterior $\bbb_j,$ i.e., 
\begin{equation}
    \begin{aligned}
        \bbb_j&\mid\mathcal{Y},\boldsymbol{\omega}_j \sim N(\bfmu_{\beta_j}, \Sigma_{\beta_j}),\\
        \Sigma_{\beta_j} &= \left[(\mathcal{Y} - \bfkappa_j\boldsymbol{e}_j^\top)^\top D_{\omega_j}(\mathcal{Y} - \bfkappa_j\boldsymbol{e}_j^\top)+D_{\beta_j}\right]^{-1},\\ \bfmu_{\beta_j} &= \Sigma_{\beta_j}(\mathcal{Y} - \bfkappa_j\boldsymbol{e}_j^\top)^\top\bfkappa_j.
    \end{aligned}
\end{equation}

\subsection[linear trans]{Derivation of posterior distribution of $\bfS$}\label{sec:derive_alpha}
Observe that there is constraint between edge parameters $s_{jk} = s_{jk}$ for $j\neq k$, which is reflected by the symmetry of matrix $\mathbf{S}$.  It is sufficient to examine the lower triangular elements of $\mathbf{S}$ while adhering to the equality constraint. To impose such symmetric constraint on $\bfS$, we reparameterize $\bfS$ by $\boldsymbol{\alpha}=\text{vech}(\bfS)$, which is the half-vectorization of $\bfS$. Specifically,
\begin{equation}
    \boldsymbol{\alpha}=(s_{11},\ldots,s_{J1},s_{22},\ldots,s_{J2},\ldots,s_{J-1,J-1},s_{J,J-1},s_{J,J})^\top = (\alpha_1,\ldots,\alpha_{J(J+1)/2})^\top.
\end{equation}
To establish a relationship between $\bfS$ and $\bfalpha$, we first define the following equation,
\begin{equation}\label{eq:col_vec}
  \bfs_j = E_j\text{vec}(\bfS).
\end{equation}
In this equation, $\text{vec}(\bfS)=(\bfs_1^\top,\ldots,\bfs_J^\top)^\top$ represents the vectorization of the matrix $\bfS$. The matrix $E_j = (0_J,\ldots,I_J,\ldots,0_J)$ is a $J\times J^2$ matrix, where the $j$th row block is the identity matrix $I_J$ and all other row blocks are zero matrices.
Next, we can express $\text{vec}(\bfS)$ as follows,
\begin{equation}\label{eq:vec_half_vec}
  \text{vec}(\bfS) = D_J\bfalpha,
\end{equation}
where $D_J$ is a $J^2\times J(J+1)/2$ duplication matrix, which can be explicitly defined as,
\begin{equation}
  D_J^\top = \sum_{i\geq j} u_{ij}(\text{vec} T_{ij})^\top.
\end{equation}
Here, $u_{ij}$ is a unit vector of order $J(J+1)/2$ with ones in the position $(j-1)J + i - j(j-1)/2$ and zeros elsewhere. The matrix $T_{ij}$ is a $J\times J$ matrix with ones in position $(i,j)$ and $(j,i)$ and zeros in all other positions.
By combining equations \eqref{eq:col_vec} and \eqref{eq:vec_half_vec}, we obtain,
\begin{equation}\label{eq:col_half_vec}
  \boldsymbol{s}_j = T_j\boldsymbol{\alpha},
\end{equation}
where $T_j=E_jD_J$ is a $J\times J(J+1)/2$ transformation matrix.

Given $\bfalpha$, we first sample $NJ$ augmentation variables $\Omega=(\omega_{ij})_{N\times J}$ from  
\begin{equation*}
\omega_{ij}\mid \mathcal{Y},\bfalpha\sim PG\left(1,\sigma_{\omega,ij}^2\right),
\end{equation*}
where $\sigma_{\omega,ij}^2$ is the $(i,j)$th entry of $\Sigma_\omega = \mathcal{Y}\bfS - (\mathcal{Y}-\frac{1}{2}\mathbf{1}_N\mathbf{1}_J^\top)\circ \mathbf{1}_N\text{diag}(\bfS)^\top$. $\bfS=\text{vech}^{-1}(\bfalpha),$ $\text{diag}(\bfS)$ is the diagonal vector of the matrix $\bfS$, and $\circ$ is the Hadamard product. Furthermore, given the above transformation, the sampling of $\bfS$ can be done instead by sampling $\bfalpha$ from its posterior with a similar {P{\'o}lya}-Gamma augmentation procedure. Specifically, the pseudo likelihood with {P{\'o}lya}-Gamma augmentation is
\begin{equation}\label{eq:ising_model_pseudo_pg}
  \begin{aligned}
    p^*(\mathcal{Y}\mid\bfS,\Omega) &= \prod_{j=1}^J p(\mathcal{Y}_j\mid\mathcal{Y}_{-j},\bfs_j,\boldsymbol{\omega}_j)\\
    &\propto \exp\left[- \frac{1}{2}\sum_{j=1}^J \left( \bfs_j^\top(\mathcal{Y} - \bfkappa_j\boldsymbol{e}_j^\top)^\top D_{\omega_j}(\mathcal{Y} - \bfkappa_j\boldsymbol{e}_j^\top)\bfs_j -2\bfkappa_j^\top(\mathcal{Y} - \bfkappa_j\boldsymbol{e}_j^\top)\bfs_j \right)\right].
  \end{aligned}
\end{equation}
Then we have the posterior of $\bfS$
\begin{equation}\label{eq:conditional_s}
  \begin{aligned}
    p&(\bfS\mid\mathcal{Y},\Omega) \propto p^*(\mathcal{Y}\mid\bfS,\Omega)\pi(\bfS)\\
    &\propto \exp\left[- \frac{1}{2}\sum_{j=1}^J \left( \bbb_j^\top(\mathcal{Y} - \bfkappa_j\boldsymbol{e}_j^\top)^\top D_{\omega_j}(\mathcal{Y} - \bfkappa_j\boldsymbol{e}_j^\top)\bbb_j -2\bfkappa_j^\top(\mathcal{Y} - \bfkappa_j\boldsymbol{e}_j^\top)\bbb_j \right)-\frac{1}{2}\sum_{j=1}^J\bbb_j^\top D_{s_j}\bbb_j\right]\\
    &= \exp\left\{-\frac{1}{2}\sum_{j=1}^J\left[\bbb_j^\top((\mathcal{Y} - \bfkappa_j\boldsymbol{e}_j^\top)^\top D_{\omega_j}(\mathcal{Y} - \bfkappa_j\boldsymbol{e}_j^\top)+D_{s_j})\bbb_j - 2\bfkappa_j^\top(\mathcal{Y} - \bfkappa_j\boldsymbol{e}_j^\top)\bbb_j\right]\right\}.
  \end{aligned}
\end{equation}
Plugging \eqref{eq:col_half_vec} into \eqref{eq:conditional_s} we have,
\begin{equation}
  \begin{aligned}
    p(\bfalpha&\mid\mathcal{Y},\Omega) \\
    &\propto \exp\left\{-\frac{1}{2}\sum_{j=1}^J\left[\bfalpha^\top T_j^\top((\mathcal{Y} - \bfkappa_j\boldsymbol{e}_j^\top)^\top D_{\omega_j}(\mathcal{Y} - \bfkappa_j\boldsymbol{e}_j^\top)+D_{s_j})T_j\bfalpha - 2\bfkappa_j^\top(\mathcal{Y} - \bfkappa_j\boldsymbol{e}_j^\top)T_j\bfalpha\right]\right\}\\
    &\propto\exp\Bigg\{-\frac{1}{2}\Bigg[\bfalpha^\top\left(\sum_{j=1}^JT_j^\top((\mathcal{Y} - \bfkappa_j\boldsymbol{e}_j^\top)^\top D_{\omega_j}(\mathcal{Y} - \bfkappa_j\boldsymbol{e}_j^\top)+D_{s_j})T_j\right)\bfalpha\\
    &\quad -2\left(\sum_{j=1}^J((\mathcal{Y} - \bfkappa_j\boldsymbol{e}_j^\top)T_j)^\top\bfkappa_j\right)^\top\bfalpha\Bigg]\Bigg\}\\
    &\propto\exp\left[-\frac{1}{2}(\bfalpha-\bfmu_\alpha)^\top\Sigma_\alpha^{-1}(\bfalpha-\bfmu_\alpha) \right],
  \end{aligned}
\end{equation}
where
\begin{equation}
  \Sigma_\alpha = \left[\sum_{j=1}^J T_j^\top((\mathcal{Y} - \bfkappa_j\boldsymbol{e}_j^\top)^\top D_{\omega_j}(\mathcal{Y} - \bfkappa_j\boldsymbol{e}_j^\top)+D_{s_j})T_j\right]^{-1}, \boldsymbol{\mu}_\alpha = \Sigma_\alpha\left[\sum_{j=1}^J ((\mathcal{Y} - \bfkappa_j\boldsymbol{e}_j^\top)T_j)^\top\boldsymbol{\kappa}_j\right],
\end{equation}
which can be further simplified as below.

In summary, the posterior of $\bfalpha$ is 
\begin{equation}\label{eq:cond_alpha_pg}
    \begin{aligned}
        \bfalpha&\mid\mathcal{Y},D_{\omega}\sim N(\boldsymbol{\mu}_\alpha,\Sigma_\alpha),\\
        \Sigma_\alpha &= \left[ M^\top D_{\omega} M + T^\top D_S T\right]^{-1},\\
        \boldsymbol{\mu}_\alpha&= \Sigma_\alpha M^\top\bfkappa,
    \end{aligned}
\end{equation}
where $\mathcal{Y} = (\bfy_1,\ldots,\bfy_N)^\top,$ $M = ([(\mathcal{Y} - \bfkappa_1\boldsymbol{e}_1^\top) T_1]^\top,\ldots,[(\mathcal{Y} - \bfkappa_j\boldsymbol{e}_j^\top) T_J]^\top)^\top$, $D_\omega = \text{diag}(\boldsymbol{\omega})$, $\boldsymbol{\omega}=(\omega_{11},\ldots,\omega_{N1},\omega_{12},\ldots,\omega_{NJ})^\top,$ $T = (T_1^\top,\ldots,T_J^\top)^\top$, $\bfkappa=(\bfkappa_1^\top,\ldots,\bfkappa_J^\top)^\top$, and $D_S = \text{diag}(\boldsymbol{\tau}),$  where $\boldsymbol{\tau}=(\tau_{11},\ldots,\tau_{J1},\tau_{12},\ldots,\tau_{JJ})^\top$, $\tau_{jl}=\sigma_1^{-2},$ for $l\neq j$ and $\tau_{jj}=\sigma_2^{-2}.$ 

Instead of conventional matrix inversion for calculating $\Sigma_\alpha$, we use Cholesky decomposition of a symmetric positive definite (SPD) matrix, which offers enhanced efficiency and numerical stability \citep[]{blake2023fast}. More precisely, we start by performing the Cholesky decomposition of $\Sigma_\alpha^{-1} = LL^\top$, and then proceed to solve two triangular systems: i) $LY = I$, and ii) $L^\top X = Y$, thus deriving $X = \Sigma_\alpha$.

\section{Detailed settings for simulations}\label{sup:simu1_s}

\subsection{True parameters used in simulation Study I}\label{app:true_s_simu1}
\begin{table}[H]
\centering
\begin{tabular}{c|cccccc}
\toprule
 $\mathbf{S}$ & node 1 & node 2 & node 3 & node 4 & node 5 & node 6 \\
\midrule
node 1 & 0.000 & - & - & - & - & - \\
node 2 & -0.737 & 0.000 & - & - & - & - \\
node 3 & 0.000 & -0.408 & 0.000 & - & - & - \\
node 4 & 0.000 & 0.000 & 0.619 & 0.000 & - & - \\
node 5 & 0.769 & 0.000 & 0.000 & 0.000 & 0.000 & - \\
node 6 & 0.791 & 0.000 & 0.000 & 0.000 & 0.741 & 0.000 \\
\bottomrule
\end{tabular}
\caption{The true parameters used in the simulation Study I.}
\label{tab:simu1_s}
\end{table}

\subsection{MAR settings in simulation Study I}\label{app:true_mis_simu1}
\begin{table}[H]
\centering
\begin{tabular}{c|ccccc}
\toprule
 $p(z_{ij}=0\mid y_{i6})$ & $z_{i1}$ & $z_{i2}$ & $z_{i3}$ & $z_{i4}$ & $z_{i5}$  \\
\midrule
$y_{i6}=1$ & 0.2 & 0.0 & 0.1 & 0.1 & 0.2 \\
$y_{i6}=0$ & 0.0 & 0.1 & 0.1 & 0.0 & 0.3 \\
\bottomrule
\end{tabular}
\caption{The probabilities of missingness for the first five variables conditioning on the sixth variable.}
\label{tab:simu1_mis}
\end{table}

\subsection{True parameters used in simulation Study II}\label{app:true_s_simu2}
\begin{table}[H]
  \centering
  \label{tab:NewIsingModelParameters}
  \begin{tabular}{@{}c|cccccc@{}}
  \toprule
  $\mathbf{S}$ & node 1 & node 2 & node 3 & node 4 & node 5 & node 6 \\
  \midrule
node 1 & 0.000 & - & - & - & - & -\\
node 2 & 0.500 & 0.000 & - & - & - & -\\
node 3 & 0.000 & 0.514 & 0.000 & - & - & -\\
node 4 & 0.000 & 0.000 & 0.865 & 0.000 & - & -\\
node 5 & 1.115 & 0.000 & 0.000 & 0.000 & 0.000 & -\\
node 6 & 1.151 & 0.000 & 0.000 & 0.000 & 1.068 & 0.000\\
  \bottomrule
  \end{tabular}
  \caption{The true parameters used in the simulation Study II.}
\end{table}


\section{Study III: A fifteen-node example.}\label{sup:add_simu}

We further simulate a  fifteen-node Ising model to demonstrate the performance of the proposed method in terms of parameter estimation and edge selection.
Similar to the six-node scenario, we generate model parameters by randomly setting 70\% of the edge parameters $s_{jl}$ to zero to create a sparse network. The true parameter values can be found in the Table~\ref{tab:simu3_s}. 
To simulate missing data, we implement the MCAR mechanism and randomly label 50\% of the data entries as missing.
Data are generated for four sample sizes of $N=$ 1,000, 2,000, 4,000, and 8,000, following the specified Ising model parameters and MCAR mechanism. For each sample size, 50 independent replications are generated.
Algorithm~\ref{alg:pseudo_imputation} is applied to these 
datasets, where a random starting point is used as in the six-node example.  We set the MCMC iterations length to $T=$ 5,000 and a burn-in size $T_0=$ 1,000.

The resulting MSEs for edge parameter estimation under various sample sizes are displayed in Figure \ref{fig:sim2}(a), where each box corresponds to the MSEs for 105 edge parameters. 
As we can see, the MSEs decrease toward zero as the sample size increases. Furthermore, by employing a hard thresholding step after edge parameter estimation, a receiver operating characteristic (ROC) curve  is created for edge selection under each setting, and the corresponding Area Under the Curve (AUC) is calculated to evaluate the performance of the proposed method.  That is, given a hard threshold $\tau$, the True Positive Rate (TPR) and False Positive Rate (FPR) are calculated as  
\begin{equation}
    \text{TPR}(\tau) = \frac{\sum_{k=1}^{50}\sum_{j<l} \mathbf{1}_{\{\vert \hat{s}_{k,jl}\vert>\tau \text{ and } s_{0,jl}\neq 0\}}}{50\times \sum_{j<l}\mathbf{1}_{\{s_{0,jl}\neq 0\}}},\quad \text{FPR}(\tau) = \frac{\sum_{k=1}^{50} \sum_{j<l}\mathbf{1}_{\{\vert \hat{s}_{k,jl}\vert>\tau \text{ and } s_{0,jl}=0\}}}{50\times\sum_{j<l}\mathbf{1}_{\{s_{0,jl}=0\} }}.
\end{equation}
A ROC curve is obtained by varying the value of $\tau$. 
The ROC curves and the corresponding AUC values are given in Figure~\ref{fig:sim2}(b). The curves of TPR and FPR varies with threshold are displayed in Figure~\ref{fig:sim2-1}(a) and Figure~\ref{fig:sim2-1}(b), respectively.
Furthermore, we show scatter plots of the estimated edge parameters against the true values for a single replication in Figure~\ref{fig:sim2-2}, which indicate that the proposed method provides a good estimation of the edge parameters. Figure \ref{fig:sim2-3} illustrates the similarity between the estimated network against the true network, measured by the Jaccard coefficient. Given a threshold $\tau$, the Jaccard coefficient is defined by
\begin{equation}
 \text{Jacc}(\tau) = \frac{\sum_{k=1}^{50}\sum_{j<l} \mathbf{1}_{\{\vert \hat{s}_{k,jl}\vert>\tau \text{ and } s_{0,jl}\neq 0\}}}{\sum_{k=1}^{50}\sum_{j<l} \mathbf{1}_{\{\vert \hat{s}_{k,jl}\vert>\tau \text{ or } s_{0,jl}\neq 0\}}}.
\end{equation}
The Jaccard coefficient shows good consistency of the estimated edges with the true edges. Finally, we show the true versus estimated networks for a single replication in Figure \ref{fig:sim2-4}, further validates that the proposed method can accurately recover the true network structure.
 
\begin{table}[H]
  \centering
  \footnotesize
  \begin{tabular}{c|ccccccccccccccc}
  \toprule
  $\mathbf{S}$ & 1 & 2 & 3 & 4 & 5 & 6 & 7 & 8 & 9 & 10 & 11 & 12 & 13 & 14 & 15\\
  \midrule
  1 & -  &  &  &  &  &  &  &  &  &  &  &  &  &  & \\
  2 & -  & -  &  &  &  &  &  &  &  &  &  &  &  &  & \\
  3 & -  & -0.96 & -  &  &  &  &  &  &  &  &  &  &  &  & \\
  4 & -  & 1.00 & 0.5 & -  &  &  &  &  &  &  &  &  &  &  & \\
  5 & 0.48 & -  & -  & -  & -  &  &  &  &  &  &  &  &  &  & \\
  6 & -  & -  & 0.95 & -  & 0.47 & -  &  &  &  &  &  &  &  &  & \\
  7 & -  & -  & -  & 0.55 & -0.92 & -  & -  &  &  &  &  &  &  &  & \\
  8 & -  & -  & -  & -  & -  & 0.98 & 0.74 & -  &  &  &  &  &  &  & \\
  9 & -0.41 & -  & -0.54 & -  & -  & -  & -  & -  & -  &  &  &  &  &  & \\
  10 & -  & -  & -0.47 & -  & -  & -  & -0.83 & -  & -0.41 & -  &  &  &  &  & \\
  11 & -  & -  & -0.74 & 0.52 & -0.55 & 0.85 & -  & 0.75 & -0.98 & -  & -  &  &  &  & \\
  12 & -0.54 & -  & 0.77 & -  & -  & 0.41 & -  & -  & -  & -0.74 & -  & -  &  &  & \\
  13 & -  & -  & -  & -  & -  & -0.8 & 0.56 & -  & -  & -0.78 & -  & -  & -  &  & \\
  14 & -  & -  & -  & -  & -  & -  & 0.95 & -  & -  & -  & -  & -  & -0.96 & -  & \\
  15 & -  & -  & -  & -  & -0.97 & -  & -  & 0.78 & -  & -  & -  & -  & -  & -  & - \\
  \bottomrule
  \end{tabular}
  \caption{The true parameters used in the 15-node Ising model simulation.}
  \label{tab:simu3_s}
  \end{table}

\begin{figure}[H]
\centering
\begin{subfigure}{.49\textwidth}
  \centering
  \includegraphics[width=.99\linewidth]{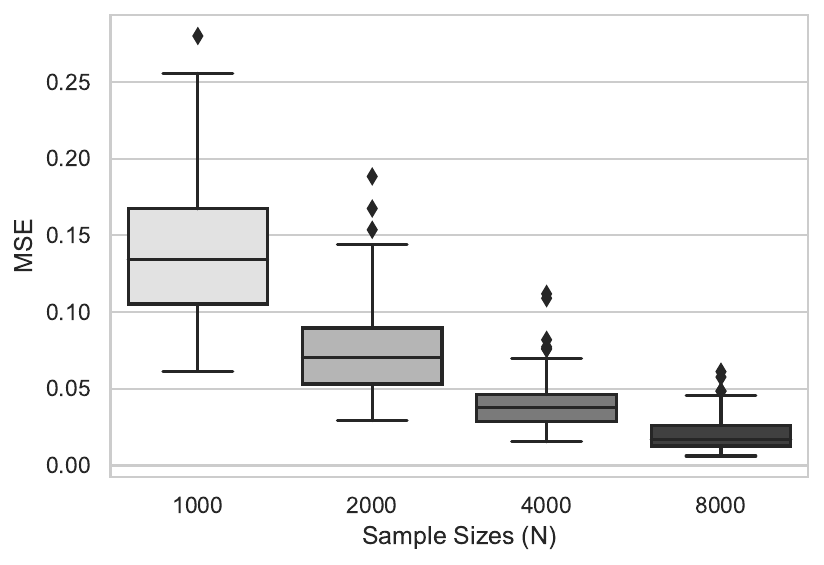}
  \caption{}
\end{subfigure}%
\begin{subfigure}{.49\textwidth}
  \centering
  \includegraphics[width=.99\linewidth]{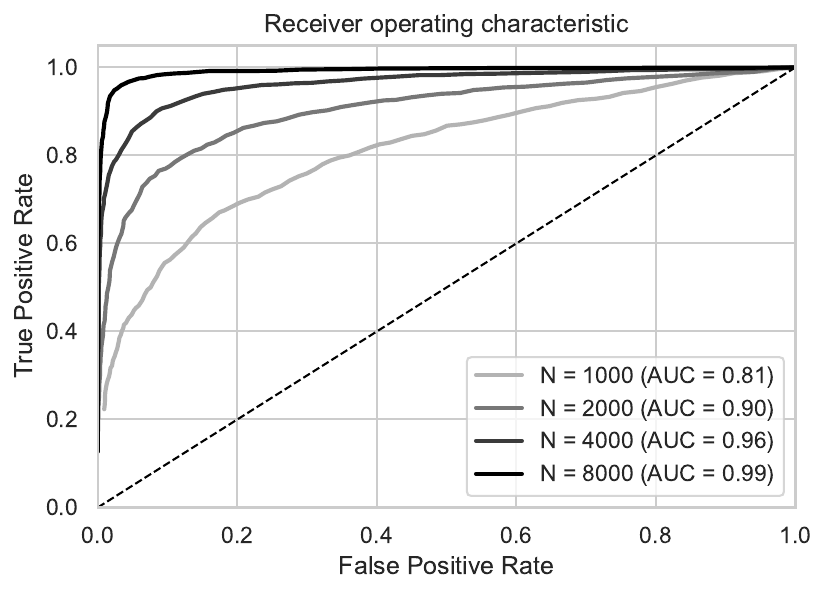}
  \caption{}
\end{subfigure}
\caption{Panel (a): Boxplots of MSEs for edge parameters. Panel (b): The ROC curves for edge selection and the corresponding AUCs.}
\label{fig:sim2}
\end{figure}

\begin{figure}[H]
  \centering
  \begin{subfigure}{.49\textwidth}
    \centering
    \includegraphics[width=.99\linewidth]{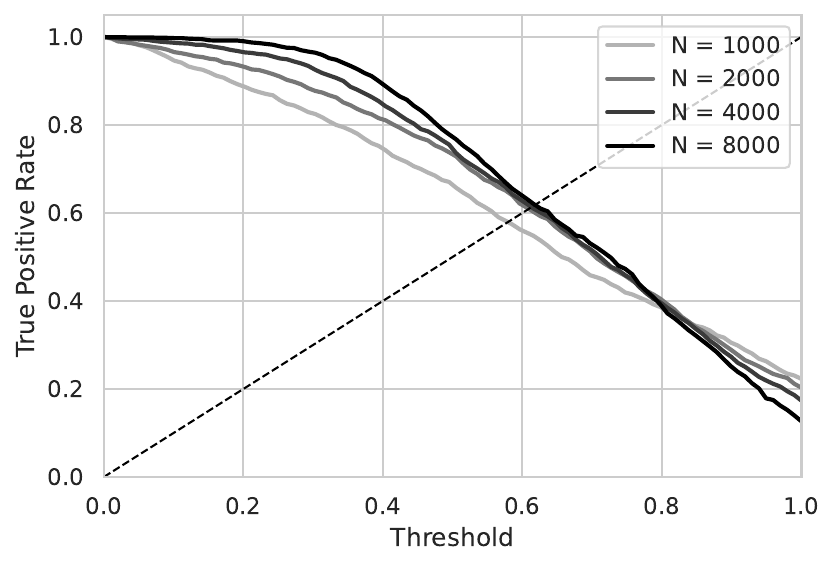}
    \caption{}
  \end{subfigure}%
  \begin{subfigure}{.49\textwidth}
    \centering
    \includegraphics[width=.99\linewidth]{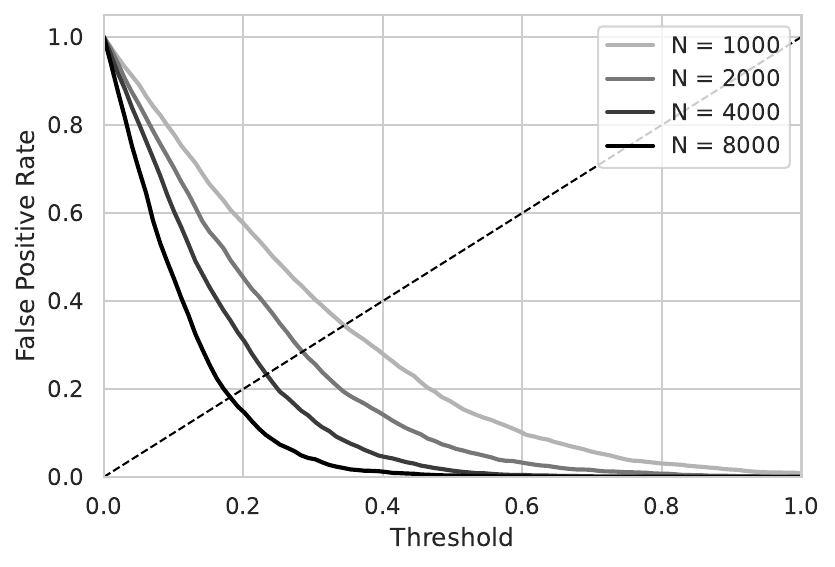}
    \caption{}
  \end{subfigure}
  \caption{Panel (a): True positive rate for edge selection. Panel (b): False positive rate for edge selection.}
  \label{fig:sim2-1}
  \end{figure}

\begin{figure}[H]
  \centering
  \begin{subfigure}{.49\textwidth}
    \centering
    \includegraphics[width=.99\linewidth]{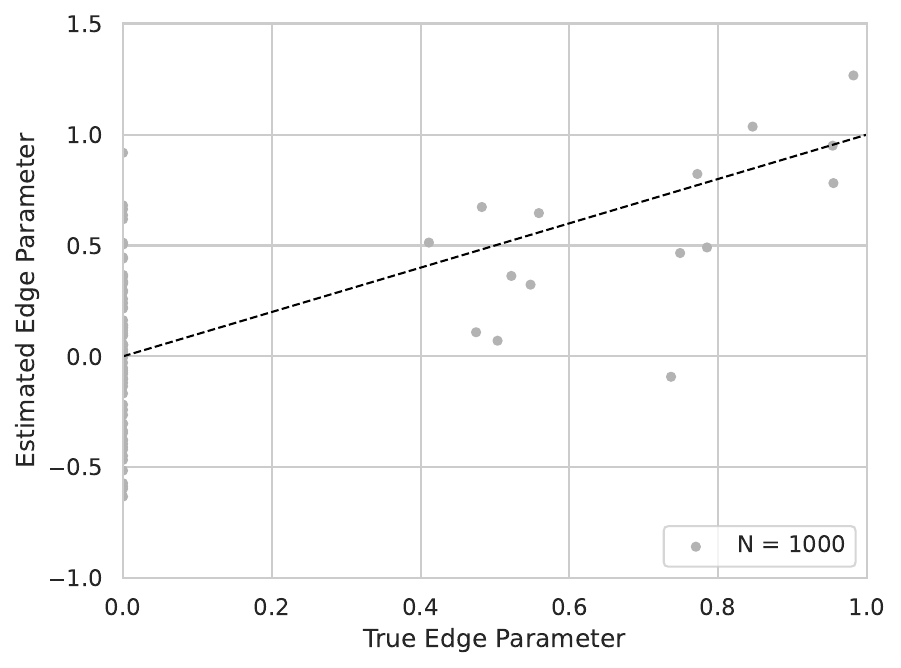}
    \caption{}
  \end{subfigure}%
  \begin{subfigure}{.49\textwidth}
    \centering
    \includegraphics[width=.99\linewidth]{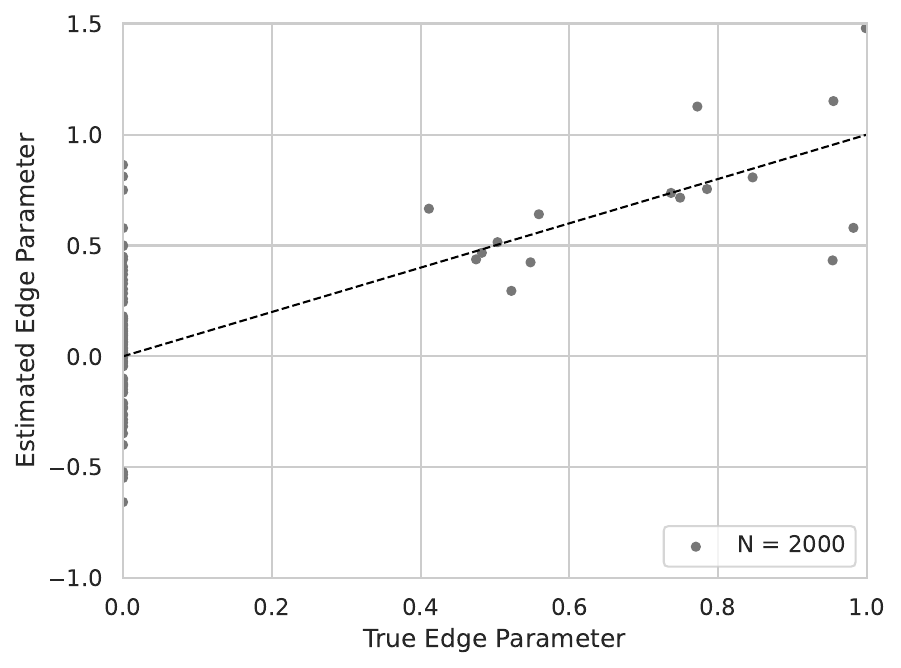}
    \caption{}
  \end{subfigure}
  \\
  \begin{subfigure}{.49\textwidth}
    \centering
    \includegraphics[width=.99\linewidth]{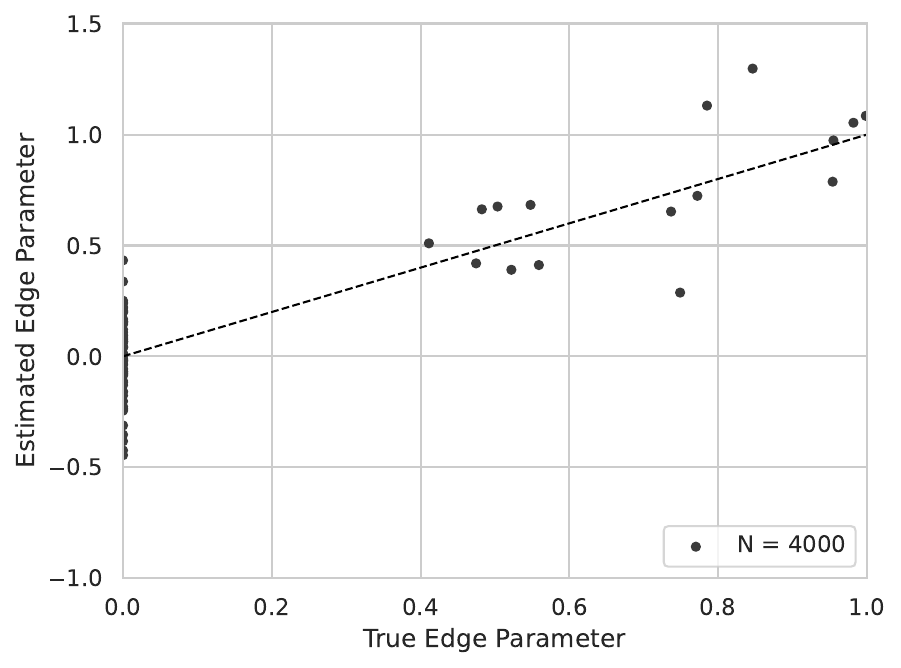}
    \caption{}
  \end{subfigure}
  \begin{subfigure}{.49\textwidth}
    \centering
    \includegraphics[width=.99\linewidth]{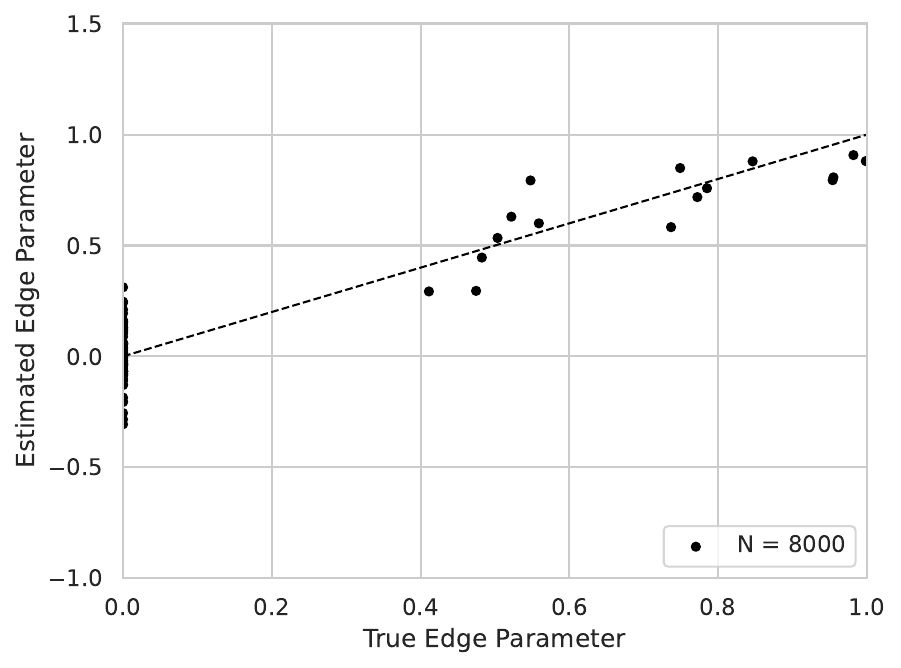}
    \caption{}
  \end{subfigure}
  \caption{Scatter plots of edge parameters across different sample sizes.}
  \label{fig:sim2-2}
  \end{figure}

\begin{figure}[H]
    \centering
    \includegraphics[width=.5\linewidth]{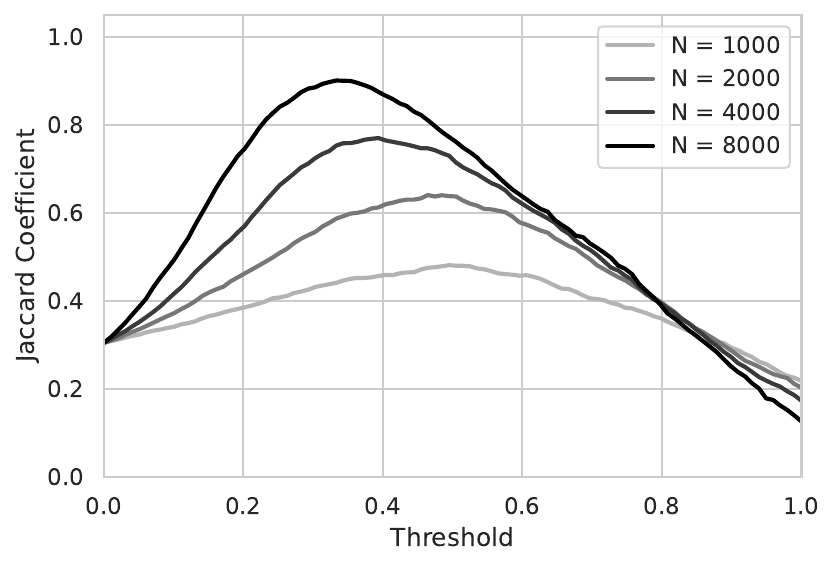}
  \caption{The Jaccard coefficients for similarity between estimated and true edges.}
  \label{fig:sim2-3}
  \end{figure}

\begin{figure}[H]
  \centering
  \begin{subfigure}{.5\textwidth}
    \centering
    \includegraphics[width=.99\linewidth]{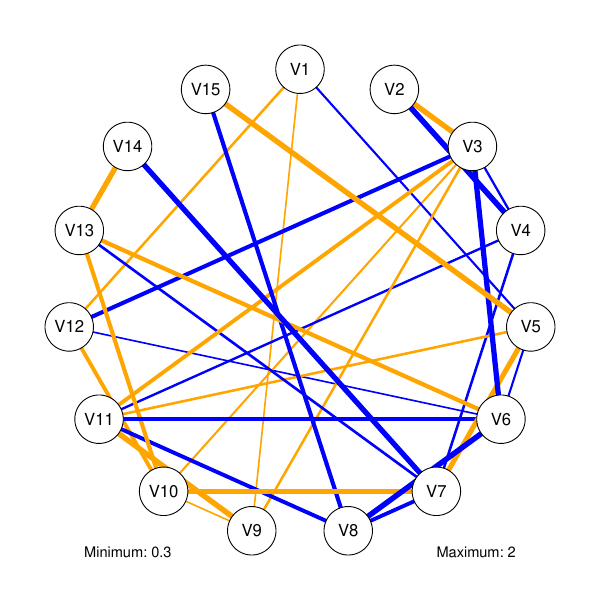}
    \caption{ }
  \end{subfigure}%
  \begin{subfigure}{.5\textwidth}
    \centering
    \includegraphics[width=.99\linewidth]{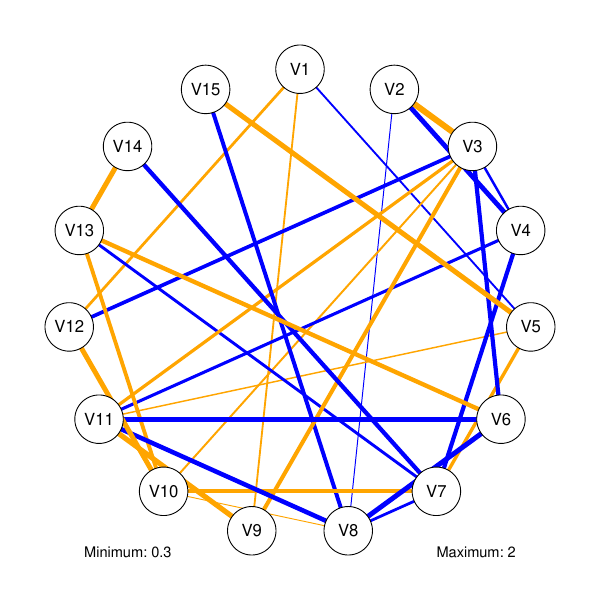}
    \caption{ }
  \end{subfigure}
  \caption{Estimated network structure for Simulation Study III. Panel (a): True network. Panel (b) Estimated network. Blue for positive edges and orange for negative edges.}
  \label{fig:sim2-4}
\end{figure}

\bibliographystyle{apacite}
\bibliography{ref}

\begin{thebibliography}{}

\bibitem [\protect \citeauthoryear {%
{American Psychiatric Association}%
}{%
{American Psychiatric Association}%
}{%
{\protect \APACyear {2000}}%
}]{%
APA2000}
\APACinsertmetastar {%
APA2000}%
\begin{APACrefauthors}%
{American Psychiatric Association}.%
\end{APACrefauthors}%
\unskip\
\newblock
\APACrefYear{2000}.
\newblock
\APACrefbtitle {Diagnostic and Statistical Manual of Mental Disorders: {DSM-IV-TR}} {Diagnostic and statistical manual of mental disorders: {DSM-IV-TR}}\ (\PrintOrdinal{4th, Text Revision}\ \BEd).
\newblock
\APACaddressPublisher{Washington, DC}{American Psychiatric Association}.
\PrintBackRefs{\CurrentBib}

\bibitem [\protect \citeauthoryear {%
Anderson%
\ \BBA {} Yu%
}{%
Anderson%
\ \BBA {} Yu%
}{%
{\protect \APACyear {2007}}%
}]{%
anderson2007log}
\APACinsertmetastar {%
anderson2007log}%
\begin{APACrefauthors}%
Anderson, C\BPBI J.%
\BCBT {}\ \BBA {} Yu, H\BHBI T.%
\end{APACrefauthors}%
\unskip\
\newblock
\APACrefYearMonthDay{2007}{}{}.
\newblock
{\BBOQ}\APACrefatitle {Log-multiplicative association models as item response models} {Log-multiplicative association models as item response models}.{\BBCQ}
\newblock
\APACjournalVolNumPages{Psychometrika}{72}{1}{5--23}.
\PrintBackRefs{\CurrentBib}

\bibitem [\protect \citeauthoryear {%
Armour%
, Fried%
, Deserno%
, Tsai%
\BCBL {}\ \BBA {} Pietrzak%
}{%
Armour%
\ \protect \BOthers {.}}{%
{\protect \APACyear {2017}}%
}]{%
armour2017network}
\APACinsertmetastar {%
armour2017network}%
\begin{APACrefauthors}%
Armour, C.%
, Fried, E\BPBI I.%
, Deserno, M\BPBI K.%
, Tsai, J.%
\BCBL {}\ \BBA {} Pietrzak, R\BPBI H.%
\end{APACrefauthors}%
\unskip\
\newblock
\APACrefYearMonthDay{2017}{}{}.
\newblock
{\BBOQ}\APACrefatitle {A network analysis of {DSM-5} posttraumatic stress disorder symptoms and correlates in {U.S.} military veterans} {A network analysis of {DSM-5} posttraumatic stress disorder symptoms and correlates in {U.S.} military veterans}.{\BBCQ}
\newblock
\APACjournalVolNumPages{Journal of Anxiety Disorders}{45}{}{49--59}.
\PrintBackRefs{\CurrentBib}

\bibitem [\protect \citeauthoryear {%
Barndorff-Nielsen%
, Kent%
\BCBL {}\ \BBA {} S{\o}rensen%
}{%
Barndorff-Nielsen%
\ \protect \BOthers {.}}{%
{\protect \APACyear {1982}}%
}]{%
barndorff1982normal}
\APACinsertmetastar {%
barndorff1982normal}%
\begin{APACrefauthors}%
Barndorff-Nielsen, O.%
, Kent, J.%
\BCBL {}\ \BBA {} S{\o}rensen, M.%
\end{APACrefauthors}%
\unskip\
\newblock
\APACrefYearMonthDay{1982}{}{}.
\newblock
{\BBOQ}\APACrefatitle {Normal variance-mean mixtures and z distributions} {Normal variance-mean mixtures and z distributions}.{\BBCQ}
\newblock
\APACjournalVolNumPages{International Statistical Review/Revue Internationale de Statistique}{50}{2}{145--159}.
\PrintBackRefs{\CurrentBib}

\bibitem [\protect \citeauthoryear {%
Besag%
}{%
Besag%
}{%
{\protect \APACyear {1974}}%
}]{%
besag1974spatial}
\APACinsertmetastar {%
besag1974spatial}%
\begin{APACrefauthors}%
Besag, J.%
\end{APACrefauthors}%
\unskip\
\newblock
\APACrefYearMonthDay{1974}{}{}.
\newblock
{\BBOQ}\APACrefatitle {Spatial interaction and the statistical analysis of lattice systems} {Spatial interaction and the statistical analysis of lattice systems}.{\BBCQ}
\newblock
\APACjournalVolNumPages{Journal of the Royal Statistical Society: Series B (Methodological)}{36}{2}{192--225}.
\PrintBackRefs{\CurrentBib}

\bibitem [\protect \citeauthoryear {%
Besag%
}{%
Besag%
}{%
{\protect \APACyear {1975}}%
}]{%
besag1975statistical}
\APACinsertmetastar {%
besag1975statistical}%
\begin{APACrefauthors}%
Besag, J.%
\end{APACrefauthors}%
\unskip\
\newblock
\APACrefYearMonthDay{1975}{}{}.
\newblock
{\BBOQ}\APACrefatitle {Statistical analysis of non-lattice data} {Statistical analysis of non-lattice data}.{\BBCQ}
\newblock
\APACjournalVolNumPages{Journal of the Royal Statistical Society: Series D (The Statistician)}{24}{3}{179--195}.
\PrintBackRefs{\CurrentBib}

\bibitem [\protect \citeauthoryear {%
Blake%
}{%
Blake%
}{%
{\protect \APACyear {2015}}%
}]{%
blake2023fast}
\APACinsertmetastar {%
blake2023fast}%
\begin{APACrefauthors}%
Blake, E.%
\end{APACrefauthors}%
\unskip\
\newblock
\APACrefYearMonthDay{2015}{}{}.
\newblock
\APACrefbtitle {Fast and Accurate Symmetric Positive Definite Matrix Inverse Using {C}holesky Decomposition.} {Fast and accurate symmetric positive definite matrix inverse using {C}holesky decomposition.}
\newblock
\begin{APACrefURL} \url{https://uk.mathworks.com/matlabcentral/fileexchange/34511-fast-and-accurate-symmetric-positive-definite-matrix-inverse-using-cholesky-decomposition} \end{APACrefURL}
\newblock
\APACrefnote{Accessed on June 11, 2023}
\PrintBackRefs{\CurrentBib}

\bibitem [\protect \citeauthoryear {%
Blanco%
\ \protect \BOthers {.}}{%
Blanco%
\ \protect \BOthers {.}}{%
{\protect \APACyear {2014}}%
}]{%
blanco2014latent}
\APACinsertmetastar {%
blanco2014latent}%
\begin{APACrefauthors}%
Blanco, C.%
, Rubio, J\BPBI M.%
, Wall, M.%
, Secades-Villa, R.%
, Beesdo-Baum, K.%
\BCBL {}\ \BBA {} Wang, S.%
\end{APACrefauthors}%
\unskip\
\newblock
\APACrefYearMonthDay{2014}{}{}.
\newblock
{\BBOQ}\APACrefatitle {The latent structure and comorbidity patterns of generalized anxiety disorder and major depressive disorder: a national study} {The latent structure and comorbidity patterns of generalized anxiety disorder and major depressive disorder: a national study}.{\BBCQ}
\newblock
\APACjournalVolNumPages{Depression and Anxiety}{31}{3}{214--222}.
\PrintBackRefs{\CurrentBib}

\bibitem [\protect \citeauthoryear {%
Borsboom%
}{%
Borsboom%
}{%
{\protect \APACyear {2008}}%
}]{%
borsboom2008psychometric}
\APACinsertmetastar {%
borsboom2008psychometric}%
\begin{APACrefauthors}%
Borsboom, D.%
\end{APACrefauthors}%
\unskip\
\newblock
\APACrefYearMonthDay{2008}{}{}.
\newblock
{\BBOQ}\APACrefatitle {Psychometric perspectives on diagnostic systems} {Psychometric perspectives on diagnostic systems}.{\BBCQ}
\newblock
\APACjournalVolNumPages{Journal of Clinical Psychology}{64}{9}{1089--1108}.
\PrintBackRefs{\CurrentBib}

\bibitem [\protect \citeauthoryear {%
Borsboom%
}{%
Borsboom%
}{%
{\protect \APACyear {2022}}%
}]{%
borsboom2022possible}
\APACinsertmetastar {%
borsboom2022possible}%
\begin{APACrefauthors}%
Borsboom, D.%
\end{APACrefauthors}%
\unskip\
\newblock
\APACrefYearMonthDay{2022}{}{}.
\newblock
{\BBOQ}\APACrefatitle {Possible futures for network psychometrics} {Possible futures for network psychometrics}.{\BBCQ}
\newblock
\APACjournalVolNumPages{Psychometrika}{87}{1}{253--265}.
\PrintBackRefs{\CurrentBib}

\bibitem [\protect \citeauthoryear {%
Borsboom%
\ \protect \BOthers {.}}{%
Borsboom%
\ \protect \BOthers {.}}{%
{\protect \APACyear {2021}}%
}]{%
borsboom2021network}
\APACinsertmetastar {%
borsboom2021network}%
\begin{APACrefauthors}%
Borsboom, D.%
, Deserno, M\BPBI K.%
, Rhemtulla, M.%
, Epskamp, S.%
, Fried, E\BPBI I.%
, McNally, R\BPBI J.%
\BDBL {}Waldorp, L\BPBI J.%
\end{APACrefauthors}%
\unskip\
\newblock
\APACrefYearMonthDay{2021}{}{}.
\newblock
{\BBOQ}\APACrefatitle {Network analysis of multivariate data in psychological science} {Network analysis of multivariate data in psychological science}.{\BBCQ}
\newblock
\APACjournalVolNumPages{Nature Reviews Methods Primers}{1}{1}{58}.
\PrintBackRefs{\CurrentBib}

\bibitem [\protect \citeauthoryear {%
Borsboom%
\ \protect \BOthers {.}}{%
Borsboom%
\ \protect \BOthers {.}}{%
{\protect \APACyear {2017}}%
}]{%
borsboom2017false}
\APACinsertmetastar {%
borsboom2017false}%
\begin{APACrefauthors}%
Borsboom, D.%
, Fried, E\BPBI I.%
, Epskamp, S.%
, Waldorp, L\BPBI J.%
, van Borkulo, C\BPBI D.%
, van~der Maas, H\BPBI L.%
\BCBL {}\ \BBA {} Cramer, A\BPBI O.%
\end{APACrefauthors}%
\unskip\
\newblock
\APACrefYearMonthDay{2017}{}{}.
\newblock
{\BBOQ}\APACrefatitle {False alarm? A comprehensive reanalysis of “Evidence that psychopathology symptom networks have limited replicability” by {Forbes, Wright, Markon, and Krueger} (2017).} {False alarm? a comprehensive reanalysis of “evidence that psychopathology symptom networks have limited replicability” by {Forbes, Wright, Markon, and Krueger} (2017).}{\BBCQ}
\newblock
\APACjournalVolNumPages{Journal of Abnormal Psychology}{126}{7}{989–999}.
\PrintBackRefs{\CurrentBib}

\bibitem [\protect \citeauthoryear {%
Brunson%
\ \BBA {} Laubenbacher%
}{%
Brunson%
\ \BBA {} Laubenbacher%
}{%
{\protect \APACyear {2018}}%
}]{%
brunson2018applications}
\APACinsertmetastar {%
brunson2018applications}%
\begin{APACrefauthors}%
Brunson, J\BPBI C.%
\BCBT {}\ \BBA {} Laubenbacher, R\BPBI C.%
\end{APACrefauthors}%
\unskip\
\newblock
\APACrefYearMonthDay{2018}{}{}.
\newblock
{\BBOQ}\APACrefatitle {Applications of network analysis to routinely collected health care data: a systematic review} {Applications of network analysis to routinely collected health care data: a systematic review}.{\BBCQ}
\newblock
\APACjournalVolNumPages{Journal of the American Medical Informatics Association}{25}{2}{210--221}.
\PrintBackRefs{\CurrentBib}

\bibitem [\protect \citeauthoryear {%
Burger%
\ \protect \BOthers {.}}{%
Burger%
\ \protect \BOthers {.}}{%
{\protect \APACyear {2022}}%
}]{%
burger2022reporting}
\APACinsertmetastar {%
burger2022reporting}%
\begin{APACrefauthors}%
Burger, J.%
, Isvoranu, A\BPBI M.%
, Lunansky, G.%
, Haslbeck, J\BPBI M\BPBI B.%
, Epskamp, S.%
, Hoekstra, R\BPBI H\BPBI A.%
\BDBL {}Blanken, T\BPBI F.%
\end{APACrefauthors}%
\unskip\
\newblock
\APACrefYearMonthDay{2022}{}{}.
\newblock
{\BBOQ}\APACrefatitle {Reporting standards for psychological network analyses in cross-sectional data} {Reporting standards for psychological network analyses in cross-sectional data}.{\BBCQ}
\newblock
\APACjournalVolNumPages{Psychological Methods}{28}{4}{806--824}.
\PrintBackRefs{\CurrentBib}

\bibitem [\protect \citeauthoryear {%
Burgess%
\ \BBA {} Hitch%
}{%
Burgess%
\ \BBA {} Hitch%
}{%
{\protect \APACyear {1999}}%
}]{%
burgess1999memory}
\APACinsertmetastar {%
burgess1999memory}%
\begin{APACrefauthors}%
Burgess, N.%
\BCBT {}\ \BBA {} Hitch, G\BPBI J.%
\end{APACrefauthors}%
\unskip\
\newblock
\APACrefYearMonthDay{1999}{}{}.
\newblock
{\BBOQ}\APACrefatitle {Memory for serial order: A network model of the phonological loop and its timing.} {Memory for serial order: A network model of the phonological loop and its timing.}{\BBCQ}
\newblock
\APACjournalVolNumPages{Psychological Review}{106}{3}{551--581}.
\PrintBackRefs{\CurrentBib}

\bibitem [\protect \citeauthoryear {%
Chen%
, Li%
, Liu%
\BCBL {}\ \BBA {} Ying%
}{%
Chen%
\ \protect \BOthers {.}}{%
{\protect \APACyear {2018}}%
}]{%
chen2018robust}
\APACinsertmetastar {%
chen2018robust}%
\begin{APACrefauthors}%
Chen, Y.%
, Li, X.%
, Liu, J.%
\BCBL {}\ \BBA {} Ying, Z.%
\end{APACrefauthors}%
\unskip\
\newblock
\APACrefYearMonthDay{2018}{}{}.
\newblock
{\BBOQ}\APACrefatitle {Robust measurement via a fused latent and graphical item response theory model} {Robust measurement via a fused latent and graphical item response theory model}.{\BBCQ}
\newblock
\APACjournalVolNumPages{Psychometrika}{83}{3}{538--562}.
\PrintBackRefs{\CurrentBib}

\bibitem [\protect \citeauthoryear {%
De~Ron%
, Fried%
\BCBL {}\ \BBA {} Epskamp%
}{%
De~Ron%
\ \protect \BOthers {.}}{%
{\protect \APACyear {2021}}%
}]{%
de2021psychological}
\APACinsertmetastar {%
de2021psychological}%
\begin{APACrefauthors}%
De~Ron, J.%
, Fried, E\BPBI I.%
\BCBL {}\ \BBA {} Epskamp, S.%
\end{APACrefauthors}%
\unskip\
\newblock
\APACrefYearMonthDay{2021}{}{}.
\newblock
{\BBOQ}\APACrefatitle {Psychological networks in clinical populations: Investigating the consequences of Berkson's bias} {Psychological networks in clinical populations: Investigating the consequences of berkson's bias}.{\BBCQ}
\newblock
\APACjournalVolNumPages{Psychological Medicine}{51}{1}{168--176}.
\PrintBackRefs{\CurrentBib}

\bibitem [\protect \citeauthoryear {%
Epskamp%
}{%
Epskamp%
}{%
{\protect \APACyear {2020}}%
}]{%
epskamp2020psychometric}
\APACinsertmetastar {%
epskamp2020psychometric}%
\begin{APACrefauthors}%
Epskamp, S.%
\end{APACrefauthors}%
\unskip\
\newblock
\APACrefYearMonthDay{2020}{}{}.
\newblock
{\BBOQ}\APACrefatitle {Psychometric network models from time-series and panel data} {Psychometric network models from time-series and panel data}.{\BBCQ}
\newblock
\APACjournalVolNumPages{Psychometrika}{85}{1}{206--231}.
\PrintBackRefs{\CurrentBib}

\bibitem [\protect \citeauthoryear {%
Epskamp%
, Borsboom%
\BCBL {}\ \BBA {} Fried%
}{%
Epskamp%
, Borsboom%
\BCBL {}\ \BBA {} Fried%
}{%
{\protect \APACyear {2018}}%
}]{%
epskamp2018estimating}
\APACinsertmetastar {%
epskamp2018estimating}%
\begin{APACrefauthors}%
Epskamp, S.%
, Borsboom, D.%
\BCBL {}\ \BBA {} Fried, E\BPBI I.%
\end{APACrefauthors}%
\unskip\
\newblock
\APACrefYearMonthDay{2018}{}{}.
\newblock
{\BBOQ}\APACrefatitle {Estimating psychological networks and their accuracy: A tutorial paper} {Estimating psychological networks and their accuracy: A tutorial paper}.{\BBCQ}
\newblock
\APACjournalVolNumPages{Behavior Research Methods}{50}{1}{195--212}.
\PrintBackRefs{\CurrentBib}

\bibitem [\protect \citeauthoryear {%
Epskamp%
\ \BBA {} Fried%
}{%
Epskamp%
\ \BBA {} Fried%
}{%
{\protect \APACyear {2018}}%
}]{%
epskamp2018tutorial}
\APACinsertmetastar {%
epskamp2018tutorial}%
\begin{APACrefauthors}%
Epskamp, S.%
\BCBT {}\ \BBA {} Fried, E\BPBI I.%
\end{APACrefauthors}%
\unskip\
\newblock
\APACrefYearMonthDay{2018}{}{}.
\newblock
{\BBOQ}\APACrefatitle {A tutorial on regularized partial correlation networks.} {A tutorial on regularized partial correlation networks.}{\BBCQ}
\newblock
\APACjournalVolNumPages{Psychological Methods}{23}{4}{617--634}.
\PrintBackRefs{\CurrentBib}

\bibitem [\protect \citeauthoryear {%
Epskamp%
, Waldorp%
, M{\~o}ttus%
\BCBL {}\ \BBA {} Borsboom%
}{%
Epskamp%
, Waldorp%
\BCBL {}\ \protect \BOthers {.}}{%
{\protect \APACyear {2018}}%
}]{%
epskamp2018gaussian}
\APACinsertmetastar {%
epskamp2018gaussian}%
\begin{APACrefauthors}%
Epskamp, S.%
, Waldorp, L\BPBI J.%
, M{\~o}ttus, R.%
\BCBL {}\ \BBA {} Borsboom, D.%
\end{APACrefauthors}%
\unskip\
\newblock
\APACrefYearMonthDay{2018}{}{}.
\newblock
{\BBOQ}\APACrefatitle {The {G}aussian graphical model in cross-sectional and time-series data} {The {G}aussian graphical model in cross-sectional and time-series data}.{\BBCQ}
\newblock
\APACjournalVolNumPages{Multivariate Behavioral Research}{53}{4}{453--480}.
\PrintBackRefs{\CurrentBib}

\bibitem [\protect \citeauthoryear {%
Fried%
\ \protect \BOthers {.}}{%
Fried%
\ \protect \BOthers {.}}{%
{\protect \APACyear {2017}}%
}]{%
fried2017mental}
\APACinsertmetastar {%
fried2017mental}%
\begin{APACrefauthors}%
Fried, E\BPBI I.%
, van Borkulo, C\BPBI D.%
, Cramer, A\BPBI O.%
, Boschloo, L.%
, Schoevers, R\BPBI A.%
\BCBL {}\ \BBA {} Borsboom, D.%
\end{APACrefauthors}%
\unskip\
\newblock
\APACrefYearMonthDay{2017}{}{}.
\newblock
{\BBOQ}\APACrefatitle {Mental disorders as networks of problems: a review of recent insights} {Mental disorders as networks of problems: a review of recent insights}.{\BBCQ}
\newblock
\APACjournalVolNumPages{Social Psychiatry and Psychiatric Epidemiology}{52}{1}{1--10}.
\PrintBackRefs{\CurrentBib}

\bibitem [\protect \citeauthoryear {%
Fried%
\ \protect \BOthers {.}}{%
Fried%
\ \protect \BOthers {.}}{%
{\protect \APACyear {2020}}%
}]{%
fried2020using}
\APACinsertmetastar {%
fried2020using}%
\begin{APACrefauthors}%
Fried, E\BPBI I.%
, Von~Stockert, S.%
, Haslbeck, J.%
, Lamers, F.%
, Schoevers, R.%
\BCBL {}\ \BBA {} Penninx, B.%
\end{APACrefauthors}%
\unskip\
\newblock
\APACrefYearMonthDay{2020}{}{}.
\newblock
{\BBOQ}\APACrefatitle {Using network analysis to examine links between individual depressive symptoms, inflammatory markers, and covariates} {Using network analysis to examine links between individual depressive symptoms, inflammatory markers, and covariates}.{\BBCQ}
\newblock
\APACjournalVolNumPages{Psychological Medicine}{50}{16}{2682--2690}.
\PrintBackRefs{\CurrentBib}

\bibitem [\protect \citeauthoryear {%
Gile%
\ \BBA {} Handcock%
}{%
Gile%
\ \BBA {} Handcock%
}{%
{\protect \APACyear {2017}}%
}]{%
gile2017analysis}
\APACinsertmetastar {%
gile2017analysis}%
\begin{APACrefauthors}%
Gile, K\BPBI J.%
\BCBT {}\ \BBA {} Handcock, M\BPBI S.%
\end{APACrefauthors}%
\unskip\
\newblock
\APACrefYearMonthDay{2017}{}{}.
\newblock
{\BBOQ}\APACrefatitle {Analysis of networks with missing data with application to the National Longitudinal Study of Adolescent Health} {Analysis of networks with missing data with application to the national longitudinal study of adolescent health}.{\BBCQ}
\newblock
\APACjournalVolNumPages{Journal of the Royal Statistical Society: Series C (Applied Statistics)}{66}{3}{501--519}.
\PrintBackRefs{\CurrentBib}

\bibitem [\protect \citeauthoryear {%
Grant%
, Moore%
, Shepard%
\BCBL {}\ \BBA {} Kaplan%
}{%
Grant%
\ \protect \BOthers {.}}{%
{\protect \APACyear {2003}}%
}]{%
grant2003source}
\APACinsertmetastar {%
grant2003source}%
\begin{APACrefauthors}%
Grant, B\BPBI F.%
, Moore, T.%
, Shepard, J.%
\BCBL {}\ \BBA {} Kaplan, K.%
\end{APACrefauthors}%
\unskip\
\newblock
\APACrefYearMonthDay{2003}{}{}.
\newblock
{\BBOQ}\APACrefatitle {Source and accuracy statement: Wave 1 national epidemiologic survey on alcohol and related conditions {(NESARC)}} {Source and accuracy statement: Wave 1 national epidemiologic survey on alcohol and related conditions {(NESARC)}}.{\BBCQ}
\newblock
\APACjournalVolNumPages{Bethesda, MD: National Institute on Alcohol Abuse and Alcoholism}{52}{}{}.
\PrintBackRefs{\CurrentBib}

\bibitem [\protect \citeauthoryear {%
Hasin%
, Goodwin%
, Stinson%
\BCBL {}\ \BBA {} Grant%
}{%
Hasin%
\ \protect \BOthers {.}}{%
{\protect \APACyear {2005}}%
}]{%
hasin2005epidemiology}
\APACinsertmetastar {%
hasin2005epidemiology}%
\begin{APACrefauthors}%
Hasin, D\BPBI S.%
, Goodwin, R\BPBI D.%
, Stinson, F\BPBI S.%
\BCBL {}\ \BBA {} Grant, B\BPBI F.%
\end{APACrefauthors}%
\unskip\
\newblock
\APACrefYearMonthDay{2005}{}{}.
\newblock
{\BBOQ}\APACrefatitle {Epidemiology of major depressive disorder: results from the National Epidemiologic Survey on Alcoholism and Related Conditions} {Epidemiology of major depressive disorder: results from the national epidemiologic survey on alcoholism and related conditions}.{\BBCQ}
\newblock
\APACjournalVolNumPages{Archives of General Psychiatry}{62}{10}{1097--1106}.
\PrintBackRefs{\CurrentBib}

\bibitem [\protect \citeauthoryear {%
Haslbeck%
\ \BBA {} Fried%
}{%
Haslbeck%
\ \BBA {} Fried%
}{%
{\protect \APACyear {2017}}%
}]{%
haslbeck2017predictable}
\APACinsertmetastar {%
haslbeck2017predictable}%
\begin{APACrefauthors}%
Haslbeck, J\BPBI M.%
\BCBT {}\ \BBA {} Fried, E\BPBI I.%
\end{APACrefauthors}%
\unskip\
\newblock
\APACrefYearMonthDay{2017}{}{}.
\newblock
{\BBOQ}\APACrefatitle {How predictable are symptoms in psychopathological networks? {A} reanalysis of 18 published datasets} {How predictable are symptoms in psychopathological networks? {A} reanalysis of 18 published datasets}.{\BBCQ}
\newblock
\APACjournalVolNumPages{Psychological Medicine}{47}{16}{2767--2776}.
\PrintBackRefs{\CurrentBib}

\bibitem [\protect \citeauthoryear {%
Hettema%
}{%
Hettema%
}{%
{\protect \APACyear {2008}}%
}]{%
hettema2008nosologic}
\APACinsertmetastar {%
hettema2008nosologic}%
\begin{APACrefauthors}%
Hettema, J\BPBI M.%
\end{APACrefauthors}%
\unskip\
\newblock
\APACrefYearMonthDay{2008}{}{}.
\newblock
{\BBOQ}\APACrefatitle {The nosologic relationship between generalized anxiety disorder and major depression} {The nosologic relationship between generalized anxiety disorder and major depression}.{\BBCQ}
\newblock
\APACjournalVolNumPages{Depression and Anxiety}{25}{4}{300--316}.
\PrintBackRefs{\CurrentBib}

\bibitem [\protect \citeauthoryear {%
Holland%
}{%
Holland%
}{%
{\protect \APACyear {1990}}%
}]{%
holland1990dutch}
\APACinsertmetastar {%
holland1990dutch}%
\begin{APACrefauthors}%
Holland, P\BPBI W.%
\end{APACrefauthors}%
\unskip\
\newblock
\APACrefYearMonthDay{1990}{}{}.
\newblock
{\BBOQ}\APACrefatitle {The {D}utch identity: A new tool for the study of item response models} {The {D}utch identity: A new tool for the study of item response models}.{\BBCQ}
\newblock
\APACjournalVolNumPages{Psychometrika}{55}{}{5--18}.
\PrintBackRefs{\CurrentBib}

\bibitem [\protect \citeauthoryear {%
Huisman%
}{%
Huisman%
}{%
{\protect \APACyear {2009}}%
}]{%
huisman2009imputation}
\APACinsertmetastar {%
huisman2009imputation}%
\begin{APACrefauthors}%
Huisman, M.%
\end{APACrefauthors}%
\unskip\
\newblock
\APACrefYearMonthDay{2009}{}{}.
\newblock
{\BBOQ}\APACrefatitle {Imputation of missing network data: Some simple procedures} {Imputation of missing network data: Some simple procedures}.{\BBCQ}
\newblock
\APACjournalVolNumPages{Journal of Social Structure}{10}{1}{1--29}.
\PrintBackRefs{\CurrentBib}

\bibitem [\protect \citeauthoryear {%
Ip%
}{%
Ip%
}{%
{\protect \APACyear {2002}}%
}]{%
ip2002locally}
\APACinsertmetastar {%
ip2002locally}%
\begin{APACrefauthors}%
Ip, E\BPBI H.%
\end{APACrefauthors}%
\unskip\
\newblock
\APACrefYearMonthDay{2002}{}{}.
\newblock
{\BBOQ}\APACrefatitle {Locally dependent latent trait model and the {Dutch} identity revisited} {Locally dependent latent trait model and the {Dutch} identity revisited}.{\BBCQ}
\newblock
\APACjournalVolNumPages{Psychometrika}{67}{}{367--386}.
\PrintBackRefs{\CurrentBib}

\bibitem [\protect \citeauthoryear {%
Ising%
}{%
Ising%
}{%
{\protect \APACyear {1925}}%
}]{%
Ising1925}
\APACinsertmetastar {%
Ising1925}%
\begin{APACrefauthors}%
Ising, E.%
\end{APACrefauthors}%
\unskip\
\newblock
\APACrefYearMonthDay{1925}{}{}.
\newblock
{\BBOQ}\APACrefatitle {Beitrag zur Theorie des Ferromagnetismus} {Beitrag zur theorie des ferromagnetismus}.{\BBCQ}
\newblock
\APACjournalVolNumPages{Zeitschrift für Physik}{31}{1}{253--258}.
\PrintBackRefs{\CurrentBib}

\bibitem [\protect \citeauthoryear {%
Kohler%
\ \protect \BOthers {.}}{%
Kohler%
\ \protect \BOthers {.}}{%
{\protect \APACyear {2022}}%
}]{%
kohler2022using}
\APACinsertmetastar {%
kohler2022using}%
\begin{APACrefauthors}%
Kohler, K.%
, Jankowski, M\BPBI D.%
, Bashford, T.%
, Goyal, D\BPBI G.%
, Habermann, E\BPBI B.%
\BCBL {}\ \BBA {} Walker, L\BPBI E.%
\end{APACrefauthors}%
\unskip\
\newblock
\APACrefYearMonthDay{2022}{}{}.
\newblock
{\BBOQ}\APACrefatitle {Using network analysis to model the effects of the {SARS Cov2} pandemic on acute patient care within a healthcare system} {Using network analysis to model the effects of the {SARS Cov2} pandemic on acute patient care within a healthcare system}.{\BBCQ}
\newblock
\APACjournalVolNumPages{Scientific Reports}{12}{1}{10050}.
\PrintBackRefs{\CurrentBib}

\bibitem [\protect \citeauthoryear {%
Koponen%
, Asikainen%
, Viholainen%
\BCBL {}\ \BBA {} Hirvonen%
}{%
Koponen%
\ \protect \BOthers {.}}{%
{\protect \APACyear {2019}}%
}]{%
koponen2019using}
\APACinsertmetastar {%
koponen2019using}%
\begin{APACrefauthors}%
Koponen, M.%
, Asikainen, M\BPBI A.%
, Viholainen, A.%
\BCBL {}\ \BBA {} Hirvonen, P\BPBI E.%
\end{APACrefauthors}%
\unskip\
\newblock
\APACrefYearMonthDay{2019}{}{}.
\newblock
{\BBOQ}\APACrefatitle {Using network analysis methods to investigate how future teachers conceptualize the links between the domains of teacher knowledge} {Using network analysis methods to investigate how future teachers conceptualize the links between the domains of teacher knowledge}.{\BBCQ}
\newblock
\APACjournalVolNumPages{Teaching and Teacher Education}{79}{}{137--152}.
\PrintBackRefs{\CurrentBib}

\bibitem [\protect \citeauthoryear {%
Li%
, Craig%
\BCBL {}\ \BBA {} Bhadra%
}{%
Li%
\ \protect \BOthers {.}}{%
{\protect \APACyear {2019}}%
}]{%
li2019graphical}
\APACinsertmetastar {%
li2019graphical}%
\begin{APACrefauthors}%
Li, Y.%
, Craig, B\BPBI A.%
\BCBL {}\ \BBA {} Bhadra, A.%
\end{APACrefauthors}%
\unskip\
\newblock
\APACrefYearMonthDay{2019}{}{}.
\newblock
{\BBOQ}\APACrefatitle {The graphical horseshoe estimator for inverse covariance matrices} {The graphical horseshoe estimator for inverse covariance matrices}.{\BBCQ}
\newblock
\APACjournalVolNumPages{Journal of Computational and Graphical Statistics}{28}{3}{747--757}.
\PrintBackRefs{\CurrentBib}

\bibitem [\protect \citeauthoryear {%
Lin%
, Fried%
\BCBL {}\ \BBA {} Eaton%
}{%
Lin%
\ \protect \BOthers {.}}{%
{\protect \APACyear {2020}}%
}]{%
lin2020association}
\APACinsertmetastar {%
lin2020association}%
\begin{APACrefauthors}%
Lin, S\BHBI Y.%
, Fried, E\BPBI I.%
\BCBL {}\ \BBA {} Eaton, N\BPBI R.%
\end{APACrefauthors}%
\unskip\
\newblock
\APACrefYearMonthDay{2020}{}{}.
\newblock
{\BBOQ}\APACrefatitle {The association of life stress with substance use symptoms: A network analysis and replication.} {The association of life stress with substance use symptoms: A network analysis and replication.}{\BBCQ}
\newblock
\APACjournalVolNumPages{Journal of Abnormal Psychology}{129}{2}{204--214}.
\PrintBackRefs{\CurrentBib}

\bibitem [\protect \citeauthoryear {%
Little%
\ \BBA {} Rubin%
}{%
Little%
\ \BBA {} Rubin%
}{%
{\protect \APACyear {2019}}%
}]{%
little2019statistical}
\APACinsertmetastar {%
little2019statistical}%
\begin{APACrefauthors}%
Little, R\BPBI J.%
\BCBT {}\ \BBA {} Rubin, D\BPBI B.%
\end{APACrefauthors}%
\unskip\
\newblock
\APACrefYear{2019}.
\newblock
\APACrefbtitle {Statistical analysis with missing data} {Statistical analysis with missing data}\ (\BVOL~793).
\newblock
\APACaddressPublisher{}{John Wiley \& Sons}.
\PrintBackRefs{\CurrentBib}

\bibitem [\protect \citeauthoryear {%
Liu%
, Gelman%
, Hill%
, Su%
\BCBL {}\ \BBA {} Kropko%
}{%
Liu%
\ \protect \BOthers {.}}{%
{\protect \APACyear {2014}}%
}]{%
liu2014stationary}
\APACinsertmetastar {%
liu2014stationary}%
\begin{APACrefauthors}%
Liu, J.%
, Gelman, A.%
, Hill, J.%
, Su, Y\BHBI S.%
\BCBL {}\ \BBA {} Kropko, J.%
\end{APACrefauthors}%
\unskip\
\newblock
\APACrefYearMonthDay{2014}{}{}.
\newblock
{\BBOQ}\APACrefatitle {On the stationary distribution of iterative imputations} {On the stationary distribution of iterative imputations}.{\BBCQ}
\newblock
\APACjournalVolNumPages{Biometrika}{101}{1}{155--173}.
\PrintBackRefs{\CurrentBib}

\bibitem [\protect \citeauthoryear {%
Luke%
\ \BBA {} Harris%
}{%
Luke%
\ \BBA {} Harris%
}{%
{\protect \APACyear {2007}}%
}]{%
luke2007network}
\APACinsertmetastar {%
luke2007network}%
\begin{APACrefauthors}%
Luke, D\BPBI A.%
\BCBT {}\ \BBA {} Harris, J\BPBI K.%
\end{APACrefauthors}%
\unskip\
\newblock
\APACrefYearMonthDay{2007}{}{}.
\newblock
{\BBOQ}\APACrefatitle {Network analysis in public health: history, methods, and applications} {Network analysis in public health: history, methods, and applications}.{\BBCQ}
\newblock
\APACjournalVolNumPages{Annual Review of Public Health}{28}{}{69--93}.
\PrintBackRefs{\CurrentBib}

\bibitem [\protect \citeauthoryear {%
Marsman%
, Huth%
, Waldorp%
\BCBL {}\ \BBA {} Ntzoufras%
}{%
Marsman%
\ \protect \BOthers {.}}{%
{\protect \APACyear {2022}}%
}]{%
marsman2022objective}
\APACinsertmetastar {%
marsman2022objective}%
\begin{APACrefauthors}%
Marsman, M.%
, Huth, K.%
, Waldorp, L.%
\BCBL {}\ \BBA {} Ntzoufras, I.%
\end{APACrefauthors}%
\unskip\
\newblock
\APACrefYearMonthDay{2022}{}{}.
\newblock
{\BBOQ}\APACrefatitle {Objective {B}ayesian Edge Screening and Structure Selection for {I}sing Networks} {Objective {B}ayesian edge screening and structure selection for {I}sing networks}.{\BBCQ}
\newblock
\APACjournalVolNumPages{Psychometrika}{87}{1}{47--82}.
\PrintBackRefs{\CurrentBib}

\bibitem [\protect \citeauthoryear {%
Marsman%
\ \BBA {} Rhemtulla%
}{%
Marsman%
\ \BBA {} Rhemtulla%
}{%
{\protect \APACyear {2022}}%
}]{%
marsman2022guest}
\APACinsertmetastar {%
marsman2022guest}%
\begin{APACrefauthors}%
Marsman, M.%
\BCBT {}\ \BBA {} Rhemtulla, M.%
\end{APACrefauthors}%
\unskip\
\newblock
\APACrefYearMonthDay{2022}{}{}.
\newblock
{\BBOQ}\APACrefatitle {Guest Editors’ Introduction to The Special Issue “Network Psychometrics in Action”: Methodological Innovations Inspired by Empirical Problems} {Guest editors’ introduction to the special issue “network psychometrics in action”: Methodological innovations inspired by empirical problems}.{\BBCQ}
\newblock
\APACjournalVolNumPages{Psychometrika}{87}{1}{1--11}.
\PrintBackRefs{\CurrentBib}

\bibitem [\protect \citeauthoryear {%
Mazumder%
\ \BBA {} Hastie%
}{%
Mazumder%
\ \BBA {} Hastie%
}{%
{\protect \APACyear {2012}}%
}]{%
mazumder2012graphical}
\APACinsertmetastar {%
mazumder2012graphical}%
\begin{APACrefauthors}%
Mazumder, R.%
\BCBT {}\ \BBA {} Hastie, T.%
\end{APACrefauthors}%
\unskip\
\newblock
\APACrefYearMonthDay{2012}{}{}.
\newblock
{\BBOQ}\APACrefatitle {The graphical {L}asso: New insights and alternatives} {The graphical {L}asso: New insights and alternatives}.{\BBCQ}
\newblock
\APACjournalVolNumPages{Electronic Journal of Statistics}{6}{}{2125--2149}.
\PrintBackRefs{\CurrentBib}

\bibitem [\protect \citeauthoryear {%
McBride%
, van Bezooijen%
, Aggen%
, Kendler%
\BCBL {}\ \BBA {} Fried%
}{%
McBride%
\ \protect \BOthers {.}}{%
{\protect \APACyear {2023}}%
}]{%
mcbride2023quantifying}
\APACinsertmetastar {%
mcbride2023quantifying}%
\begin{APACrefauthors}%
McBride, O.%
, van Bezooijen, J.%
, Aggen, S\BPBI H.%
, Kendler, K\BPBI S.%
\BCBL {}\ \BBA {} Fried, E\BPBI I.%
\end{APACrefauthors}%
\unskip\
\newblock
\APACrefYearMonthDay{2023}{}{}.
\newblock
{\BBOQ}\APACrefatitle {Quantifying skip-out information loss when assessing major depression symptoms.} {Quantifying skip-out information loss when assessing major depression symptoms.}{\BBCQ}
\newblock
\APACjournalVolNumPages{Journal of Psychopathology and Clinical Science}{132}{4}{396--408}.
\PrintBackRefs{\CurrentBib}

\bibitem [\protect \citeauthoryear {%
Miller%
}{%
Miller%
}{%
{\protect \APACyear {2021}}%
}]{%
miller2021asymptotic}
\APACinsertmetastar {%
miller2021asymptotic}%
\begin{APACrefauthors}%
Miller, J\BPBI W.%
\end{APACrefauthors}%
\unskip\
\newblock
\APACrefYearMonthDay{2021}{}{}.
\newblock
{\BBOQ}\APACrefatitle {Asymptotic normality, concentration, and coverage of generalized posteriors} {Asymptotic normality, concentration, and coverage of generalized posteriors}.{\BBCQ}
\newblock
\APACjournalVolNumPages{The Journal of Machine Learning Research}{22}{1}{7598--7650}.
\PrintBackRefs{\CurrentBib}

\bibitem [\protect \citeauthoryear {%
Mkhitaryan%
, Crutzen%
, Steenaart%
\BCBL {}\ \BBA {} de Vries%
}{%
Mkhitaryan%
\ \protect \BOthers {.}}{%
{\protect \APACyear {2019}}%
}]{%
mkhitaryan2019network}
\APACinsertmetastar {%
mkhitaryan2019network}%
\begin{APACrefauthors}%
Mkhitaryan, S.%
, Crutzen, R.%
, Steenaart, E.%
\BCBL {}\ \BBA {} de Vries, N\BPBI K.%
\end{APACrefauthors}%
\unskip\
\newblock
\APACrefYearMonthDay{2019}{}{}.
\newblock
{\BBOQ}\APACrefatitle {Network approach in health behavior research: how can we explore new questions?} {Network approach in health behavior research: how can we explore new questions?}{\BBCQ}
\newblock
\APACjournalVolNumPages{Health Psychology and Behavioral Medicine}{7}{1}{362--384}.
\PrintBackRefs{\CurrentBib}

\bibitem [\protect \citeauthoryear {%
Noghrehchi%
, Stoklosa%
, Penev%
\BCBL {}\ \BBA {} Warton%
}{%
Noghrehchi%
\ \protect \BOthers {.}}{%
{\protect \APACyear {2021}}%
}]{%
noghrehchi2021selecting}
\APACinsertmetastar {%
noghrehchi2021selecting}%
\begin{APACrefauthors}%
Noghrehchi, F.%
, Stoklosa, J.%
, Penev, S.%
\BCBL {}\ \BBA {} Warton, D\BPBI I.%
\end{APACrefauthors}%
\unskip\
\newblock
\APACrefYearMonthDay{2021}{}{}.
\newblock
{\BBOQ}\APACrefatitle {Selecting the model for multiple imputation of missing data: Just use an {IC}!} {Selecting the model for multiple imputation of missing data: Just use an {IC}!}{\BBCQ}
\newblock
\APACjournalVolNumPages{Statistics in Medicine}{40}{10}{2467--2497}.
\PrintBackRefs{\CurrentBib}

\bibitem [\protect \citeauthoryear {%
Polson%
, Scott%
\BCBL {}\ \BBA {} Windle%
}{%
Polson%
\ \protect \BOthers {.}}{%
{\protect \APACyear {2013}}%
}]{%
polson2013bayesian}
\APACinsertmetastar {%
polson2013bayesian}%
\begin{APACrefauthors}%
Polson, N\BPBI G.%
, Scott, J\BPBI G.%
\BCBL {}\ \BBA {} Windle, J.%
\end{APACrefauthors}%
\unskip\
\newblock
\APACrefYearMonthDay{2013}{}{}.
\newblock
{\BBOQ}\APACrefatitle {Bayesian inference for logistic models using {P{\'o}lya--Gamma} latent variables} {Bayesian inference for logistic models using {P{\'o}lya--Gamma} latent variables}.{\BBCQ}
\newblock
\APACjournalVolNumPages{Journal of the American statistical Association}{108}{504}{1339--1349}.
\PrintBackRefs{\CurrentBib}

\bibitem [\protect \citeauthoryear {%
Ročkov{\'a}%
}{%
Ročkov{\'a}%
}{%
{\protect \APACyear {2018}}%
}]{%
rovckova2018bayesian}
\APACinsertmetastar {%
rovckova2018bayesian}%
\begin{APACrefauthors}%
Ročkov{\'a}, V.%
\end{APACrefauthors}%
\unskip\
\newblock
\APACrefYearMonthDay{2018}{}{}.
\newblock
{\BBOQ}\APACrefatitle {{Bayesian estimation of sparse signals with a continuous spike-and-slab prior}} {{Bayesian estimation of sparse signals with a continuous spike-and-slab prior}}.{\BBCQ}
\newblock
\APACjournalVolNumPages{The Annals of Statistics}{46}{1}{401--437}.
\PrintBackRefs{\CurrentBib}

\bibitem [\protect \citeauthoryear {%
Ryan%
, Bringmann%
\BCBL {}\ \BBA {} Schuurman%
}{%
Ryan%
\ \protect \BOthers {.}}{%
{\protect \APACyear {2022}}%
}]{%
ryan2022challenge}
\APACinsertmetastar {%
ryan2022challenge}%
\begin{APACrefauthors}%
Ryan, O.%
, Bringmann, L\BPBI F.%
\BCBL {}\ \BBA {} Schuurman, N\BPBI K.%
\end{APACrefauthors}%
\unskip\
\newblock
\APACrefYearMonthDay{2022}{}{}.
\newblock
{\BBOQ}\APACrefatitle {The challenge of generating causal hypotheses using network models} {The challenge of generating causal hypotheses using network models}.{\BBCQ}
\newblock
\APACjournalVolNumPages{Structural Equation Modeling: A Multidisciplinary Journal}{29}{6}{953--970}.
\PrintBackRefs{\CurrentBib}

\bibitem [\protect \citeauthoryear {%
Siew%
}{%
Siew%
}{%
{\protect \APACyear {2020}}%
}]{%
siew2020applications}
\APACinsertmetastar {%
siew2020applications}%
\begin{APACrefauthors}%
Siew, C\BPBI S.%
\end{APACrefauthors}%
\unskip\
\newblock
\APACrefYearMonthDay{2020}{}{}.
\newblock
{\BBOQ}\APACrefatitle {Applications of network science to education research: Quantifying knowledge and the development of expertise through network analysis} {Applications of network science to education research: Quantifying knowledge and the development of expertise through network analysis}.{\BBCQ}
\newblock
\APACjournalVolNumPages{Education Sciences}{10}{4}{1--16}.
\PrintBackRefs{\CurrentBib}

\bibitem [\protect \citeauthoryear {%
Simon~de Blas%
, Gomez~Gonzalez%
\BCBL {}\ \BBA {} Criado~Herrero%
}{%
Simon~de Blas%
\ \protect \BOthers {.}}{%
{\protect \APACyear {2021}}%
}]{%
simon2021network}
\APACinsertmetastar {%
simon2021network}%
\begin{APACrefauthors}%
Simon~de Blas, C.%
, Gomez~Gonzalez, D.%
\BCBL {}\ \BBA {} Criado~Herrero, R.%
\end{APACrefauthors}%
\unskip\
\newblock
\APACrefYearMonthDay{2021}{03}{}.
\newblock
{\BBOQ}\APACrefatitle {Network analysis: An indispensable tool for curricula design. {A} real case-study of the degree on mathematics at the {URJC} in {S}pain} {Network analysis: An indispensable tool for curricula design. {A} real case-study of the degree on mathematics at the {URJC} in {S}pain}.{\BBCQ}
\newblock
\APACjournalVolNumPages{PLOS ONE}{16}{3}{1-21}.
\PrintBackRefs{\CurrentBib}

\bibitem [\protect \citeauthoryear {%
Sweet%
, Thomas%
\BCBL {}\ \BBA {} Junker%
}{%
Sweet%
\ \protect \BOthers {.}}{%
{\protect \APACyear {2013}}%
}]{%
sweet2013hierarchical}
\APACinsertmetastar {%
sweet2013hierarchical}%
\begin{APACrefauthors}%
Sweet, T\BPBI M.%
, Thomas, A\BPBI C.%
\BCBL {}\ \BBA {} Junker, B\BPBI W.%
\end{APACrefauthors}%
\unskip\
\newblock
\APACrefYearMonthDay{2013}{}{}.
\newblock
{\BBOQ}\APACrefatitle {Hierarchical network models for education research: Hierarchical latent space models} {Hierarchical network models for education research: Hierarchical latent space models}.{\BBCQ}
\newblock
\APACjournalVolNumPages{Journal of Educational and Behavioral Statistics}{38}{3}{295--318}.
\PrintBackRefs{\CurrentBib}

\bibitem [\protect \citeauthoryear {%
Van~Borkulo%
\ \protect \BOthers {.}}{%
Van~Borkulo%
\ \protect \BOthers {.}}{%
{\protect \APACyear {2014}}%
}]{%
van2014new}
\APACinsertmetastar {%
van2014new}%
\begin{APACrefauthors}%
Van~Borkulo, C\BPBI D.%
, Borsboom, D.%
, Epskamp, S.%
, Blanken, T\BPBI F.%
, Boschloo, L.%
, Schoevers, R\BPBI A.%
\BCBL {}\ \BBA {} Waldorp, L\BPBI J.%
\end{APACrefauthors}%
\unskip\
\newblock
\APACrefYearMonthDay{2014}{}{}.
\newblock
{\BBOQ}\APACrefatitle {A new method for constructing networks from binary data} {A new method for constructing networks from binary data}.{\BBCQ}
\newblock
\APACjournalVolNumPages{Scientific Reports}{4}{1}{1--10}.
\PrintBackRefs{\CurrentBib}

\bibitem [\protect \citeauthoryear {%
van Buuren%
}{%
van Buuren%
}{%
{\protect \APACyear {2018}}%
}]{%
buurenFlexibleImputationMissing2018}
\APACinsertmetastar {%
buurenFlexibleImputationMissing2018}%
\begin{APACrefauthors}%
van Buuren, S.%
\end{APACrefauthors}%
\unskip\
\newblock
\APACrefYear{2018}.
\newblock
\APACrefbtitle {Flexible {{Imputation}} of {{Missing Data}}} {Flexible {{Imputation}} of {{Missing Data}}}\ (\PrintOrdinal{Second}\ \BEd).
\newblock
\APACaddressPublisher{{Boca Raton}}{{Chapman and Hall/CRC}}.
\PrintBackRefs{\CurrentBib}

\bibitem [\protect \citeauthoryear {%
Van Der~Maas%
, Kan%
, Marsman%
\BCBL {}\ \BBA {} Stevenson%
}{%
Van Der~Maas%
\ \protect \BOthers {.}}{%
{\protect \APACyear {2017}}%
}]{%
van2017network}
\APACinsertmetastar {%
van2017network}%
\begin{APACrefauthors}%
Van Der~Maas, H\BPBI L.%
, Kan, K\BHBI J.%
, Marsman, M.%
\BCBL {}\ \BBA {} Stevenson, C\BPBI E.%
\end{APACrefauthors}%
\unskip\
\newblock
\APACrefYearMonthDay{2017}{}{}.
\newblock
{\BBOQ}\APACrefatitle {Network models for cognitive development and intelligence} {Network models for cognitive development and intelligence}.{\BBCQ}
\newblock
\APACjournalVolNumPages{Journal of Intelligence}{5}{2}{16}.
\PrintBackRefs{\CurrentBib}

\bibitem [\protect \citeauthoryear {%
van~der Maas%
\ \protect \BOthers {.}}{%
van~der Maas%
\ \protect \BOthers {.}}{%
{\protect \APACyear {2006}}%
}]{%
vanderMaas2006}
\APACinsertmetastar {%
vanderMaas2006}%
\begin{APACrefauthors}%
van~der Maas, H\BPBI L\BPBI J.%
, Dolan, C\BPBI V.%
, Grasman, R\BPBI P\BPBI P\BPBI P.%
, Wicherts, J\BPBI M.%
, Huizenga, H\BPBI M.%
\BCBL {}\ \BBA {} Raijmakers, M\BPBI E\BPBI J.%
\end{APACrefauthors}%
\unskip\
\newblock
\APACrefYearMonthDay{2006}{}{}.
\newblock
{\BBOQ}\APACrefatitle {A dynamical model of general intelligence: The positive manifold of intelligence by mutualism} {A dynamical model of general intelligence: The positive manifold of intelligence by mutualism}.{\BBCQ}
\newblock
\APACjournalVolNumPages{Psychological Review}{113}{4}{842--861}.
\PrintBackRefs{\CurrentBib}

\bibitem [\protect \citeauthoryear {%
Willcox%
\ \BBA {} Huang%
}{%
Willcox%
\ \BBA {} Huang%
}{%
{\protect \APACyear {2017}}%
}]{%
willcox2017network}
\APACinsertmetastar {%
willcox2017network}%
\begin{APACrefauthors}%
Willcox, K\BPBI E.%
\BCBT {}\ \BBA {} Huang, L.%
\end{APACrefauthors}%
\unskip\
\newblock
\APACrefYearMonthDay{2017}{}{}.
\newblock
{\BBOQ}\APACrefatitle {Network models for mapping educational data} {Network models for mapping educational data}.{\BBCQ}
\newblock
\APACjournalVolNumPages{Design Science}{3}{}{e18}.
\PrintBackRefs{\CurrentBib}

\bibitem [\protect \citeauthoryear {%
Yuan%
\ \BBA {} Lin%
}{%
Yuan%
\ \BBA {} Lin%
}{%
{\protect \APACyear {2007}}%
}]{%
yuan2007model}
\APACinsertmetastar {%
yuan2007model}%
\begin{APACrefauthors}%
Yuan, M.%
\BCBT {}\ \BBA {} Lin, Y.%
\end{APACrefauthors}%
\unskip\
\newblock
\APACrefYearMonthDay{2007}{}{}.
\newblock
{\BBOQ}\APACrefatitle {Model selection and estimation in the {G}aussian graphical model} {Model selection and estimation in the {G}aussian graphical model}.{\BBCQ}
\newblock
\APACjournalVolNumPages{Biometrika}{94}{1}{19--35}.
\PrintBackRefs{\CurrentBib}

\end{thebibliography}

\end{document}